\documentclass{article}

\usepackage[left=1.75cm, top=2.7cm, bottom=2.7cm,right=1.75cm]{geometry}

\usepackage{amsmath, amsfonts}
\usepackage[dvips]{graphicx}
\usepackage{theorem}
\usepackage{hyperref}
\usepackage{color}

\newenvironment{proof}[1][Proof]{\textbf{#1.} }{\ \rule{0.5em}{0.5em}}

\newtheorem{thm}{Theorem}[section]

\newtheorem{cor}[thm]{Corollary}
\newtheorem{lem}[thm]{Lemma}
\newtheorem{prop}[thm]{Proposition}

{\theorembodyfont{\rmfamily}\newtheorem{rem}{Remark}[section]}

\newcommand{\be}{\begin{equation}}
\newcommand{\ee}{\end{equation}}
\newcommand{\bq}{\begin{eqnarray}}
\newcommand{\eq}{\end{eqnarray}}

\newcommand{\bs}{\bigskip}
\newcommand{\nind}{\noindent}
\newcommand{\nn}{\nonumber}

\usepackage{latexsym}
\usepackage{amssymb}

    \definecolor{Red}{rgb}{1.00, 0.00, 0.00}
    
    \definecolor{DRed}{rgb}{0.0, 0.5, 0.00}
    
    \definecolor{Blue}{rgb}{0.00, 0.00, 1.00}
    
        \definecolor{Green}{rgb}{0.00, .50, 0.00}
    
    \definecolor{PaleGrey}{rgb}{.6, .6, .6}

    \definecolor{Purple}{rgb}{.5, 0.00, 1.0}

\begin{document}

\title{The small-maturity smile for exponential L\'{e}vy models}
\author{Jos\'e E. Figueroa-L\'{o}pez\thanks{Department of Statistics, Purdue University, W. Lafayette, IN, USA ({\tt figueroa@purdue.edu}), work partially supported by the NSF grant \# DMS 0906919.}
\and Martin Forde\thanks{Department of Mathematics, KingÕs College London, Strand, London WC2R 2LS ({\tt Martin.Forde@kcl.ac.uk}).}}


\date{}

\maketitle

\begin{abstract}
We derive a small-time expansion for out-of-the-money call options under an exponential L\'{e}vy model, using the small-time expansion for the distribution function given in Figueroa-L\'{o}pez\&Houdr\'{e}\cite{FLH09}, combined with a change of num\'{e}raire via the Esscher transform. In particular, we find that the effect of a non-zero volatility $\sigma$ of the Gaussian component of the driving L\'{e}vy process is to increase the call price by $\frac{1}{2}\sigma^2 t^2 e^{k}\nu(k)(1+o(1))$ as $t \to 0$, where $\nu$ is the L\'evy density.  Using the small-time expansion for call options, we then derive a small-time expansion for the implied volatility $\hat{\sigma}_{t}^{2}(k)$ {at log-moneyness $k$}, which sharpens the first order estimate $\hat{\sigma}_{t}^{2}(k)\sim  \frac{\frac{1}{2}k^2}{t\log (1/t)}$ given in \cite{Tnkv10}.  Our numerical results show that the second order approximation can significantly outperform the first order approximation.
Our results are also extended to a class of time-changed L\'evy models.
We also consider a small-time, small log-moneyness regime for the CGMY model, and apply this approach to the small-time pricing of at-the-money call options; we show that for {$Y\in(1,2)$}, $\lim_{t\to{}0}t^{-1/Y}\mathbb{E}(S_t-S_0)_{+}=S_{0}\mathbb{E}^{*}(Z_{+})$ and the corresponding at-the-money implied volatility $\hat{\sigma}_t(0)$ satisfies $\lim_{t \to 0}\hat{\sigma}_t(0)/t^{1/Y-1/2}=\sqrt{2\pi}\,\mathbb{E}^{*}(Z_{+})$, where $Z$ is a symmetric {$Y$-stable} random variable under $\mathbb{P}^*$ and {$Y$ is the usual parameter for the CGMY model appearing in the L\'evy density $\nu(x)=C x^{-1-Y}e^{-M x}{\bf 1}_{\{x>0\}}+C |x|^{-1-Y}e^{-G|x|}{\bf 1}_{\{x<0\}}$ of the process}.
\end{abstract}

\section{Introduction}
L\'evy processes have played an important role in the development of financial models which can accurately approximate the so-called stylized features of historical asset prices and option prices.  In the ``statistical world", financial asset prices exhibit distributions with heavy tails and high kurtosis as well as other dynamical features such as volatility clustering and leverage. In the ``risk-neutral world", market prices of vanilla options exhibit ``skewed" implied volatilities (relative to changes in the strike),  contradicting the classical Black-Scholes model which predicts a flat implied volatility smile.  The smile phenomenon has been more pronounced after the 1987 market crash.  Concretely, out-of-the-money equity put options typically bear a higher risk-premium (larger implied volatilities) than in-the-money puts.  This effect is more dramatic as the time-to-maturity decreases. As explained in \cite{CT04} (see Section 1.2.2), the latter empirical fact is viewed by many as a clear indication that a jump risk is recognized by the participants in the option market, and stochastic volatility models are, in general, not able to reproduce the pronounced implied volatility skew of short-term option prices unless the ``volatility of volatility" is forced to take high values.

The literature on small-time asymptotics for option prices and implied volatilities has grown significantly during the last decade. For recent accounts of the subject in the case of stochastic volatility models, we refer the reader to \cite{GHLOW09} for local volatility models, \cite{FJL10} for the Heston model, \cite{Forde09} for a general uncorrelated local-stochastic volatility model and \cite{Forde10} for SABR type models.  We concentrate here on asset price models with jumps. For an It\^o semimartingale model for the underlying price process $(S_{t})$, Carr\&Wu\cite{CW03} argued,  {by partially heuristic arguments, that the price of an out-of-the-money call option converges to zero at sharply different speeds depending on whether the underlying asset price process is purely continuous, purely discontinuous, or a combination of both. For instance, in the presence of jumps, they argue that the
\begin{equation}\label{CrrWuClaimAymp}
	\mathbb{E}(S_t-K)^{+} - (S_{0}-K)_{+}\sim c\, t,
\end{equation}
for some constant $c\neq{}0$, as the time-to-maturity $t$ tends to $0$ \footnote{Actually, \cite{CW03} wrote $\Pi_{t}(K) - (S_{0}-K)_{+} =O(t)$, even though in their empirical analysis they are assuming a stronger statement such as (\ref{CrrWuClaimAymp}).}, while the call price converges at the rate $O(e^{-c/t})$ for a purely-continuous model}.  These {statements} were subsequently exploited in \cite{CW03} to investigate which kind of model is more adequate to describe the observed market option prices near to expiration. They concluded the necessity of both a continuous and a jump component to describe the implied volatility of S\&P 500 index options and argued, {based on simulation experiments}, that the theoretical asymptotic behavior is usually manifested by options maturing within 20 days. {We also refer the reader to \cite{AS02} for further empirical evidence on the presence of both a continuous and jump component.}

Using the closed-form expressions for call option prices,  Boyarchenko\&Levendorksii\cite{BL02} (see also \cite{Lev04a},  \cite{Lev04b},  \cite{Lev04c}) establish the following small-time asymptotic behaviour
\begin{equation}\label{ARFO}
	\frac{1}{S_0}\mathbb{E} (S_{t}-K)_{+} \,\sim \, t {\int (e^{x}-e^{k})_{+} \nu(dx)},  \quad \quad \quad \quad ({k>0 \quad \&\quad t \to 0})\,,	
\end{equation}
for several popular exponential L\'evy models {$S_{t}=S_{0}e^{X_{t}}$}, where {$k$ is the log-moneyness $k:=\log (K/S_{0})$} and $\nu$ is the L\'evy measure of the underlying L\'evy process {$(X_{t})$}.  {Subsequently, Levendorskii \cite{Lev08} obtained (\ref{ARFO}) under certain technical conditions (see Theorem 2.1 therein), namely that $\int (|x|^{2}\wedge 1)e^{x+\varepsilon|x|} \nu(dx)<\infty$ for some $\varepsilon>0$, and $\lim_{t\to{}0} \mathbb{E} (S_{t}-K)_{+}/t$ exists in the ``out-of-the-money region".}
{More recently}, Roper\cite{Rop10} and Tankov \cite{Tnkv10} prove that (\ref{ARFO}) holds for a general L\'{e}vy {process $(X_{t})$ {under mild conditions}, using the first-order small-time moment asymptotic result
\begin{equation}
	\lim_{t\to{}0}\frac{1}{t}\mathbb{E}\left\{\varphi(X_{t})\right\}=\int \varphi(x)\nu(dx),
\end{equation}
valid  for functions $\varphi$ that converges to $0$ as $x\to{}0$ at an appropriate rate (see, e.g. Figueroa-L\'opez\cite{FL08} for details).  {In particular, it suffices that $\int_{|x|\geq 1} e^{x}\nu(dx)<\infty$.}  {\cite{Lev08} also provides a natural generalization of \eqref{ARFO} for a wide class of multi-factor L\'{e}vy and Markov models.}

{As a corollary of (\ref{ARFO})}, \cite{Rop10} {and  \cite{Tnkv10} prove independently} that the implied volatility $\hat\sigma_{t}(k)$ for exponential L\'evy models explodes near expiration for out-of-the-money vanilla options. {This is a very peculiar feature of financial models with jumps (see Remark \ref{RmrkBlowing} below for a brief discussion about its meaning).} \cite{Tnkv10} goes one step further and shows that
\begin{equation}\label{ImplVolAsympBeh1}
	\hat{\sigma}_{t}^{2}(k)\sim  \frac{\frac{1}{2}k^2}{t\log (1/t)}
\end{equation}
as $t\to{}0$.  For at-the-money call option prices, \cite{Rop10} also shows that the leading order term is $O(\sqrt{t})$ and does not depend on the jump component of the model.  Moreover, the at-the-money implied volatility converges to the volatility of the Gaussian component of the driving L\'evy process, and the limit is zero if the L\'evy process has no Gaussian part. For bounded variation L\'evy processes and for certain tempered-stable like L\'evy processes, \cite{Tnkv10} gives also the first-order asymptotic behavior of {at-the-money} implied volatilities. { The asymptotic behavior (\ref{ImplVolAsympBeh1}) is in sharp contrast with a pure-continuous stochastic volatility model, where the implied volatility converges to a non-negative constant which depends on the shortest distance {from zero to the vertical line with $x$-coordinate equal to the log-moneyness of the call option}, under the Riemmanian metric induced by the diffusion coefficient for the model (see, e.g. \cite{GHLOW09}, \cite{FJL10}, \cite{Forde09}, \cite{Forde10}).}

In this article, we extend previous results by computing the second order correction term {$a_{1}(k)$ in the call option price approximation:
\begin{equation}\label{PureLevyApprx}
	\frac{1}{t}\,\frac{1}{S_0}\mathbb{E}(S_t-K)_{+}\,=\,a_{0}(k)\,+\,
  a_{1}(k)\,t+o(t) \quad \quad \quad \quad (t \to 0)\,.
\end{equation}
}An important component in our proofs is played by the recent higher order small-time expansions for the distribution function of a L\'evy process obtained in Figueroa-L\'{o}pez\&Houdr\'{e}\cite{FLH09}.
 In the spirit of the Black-Scholes formula and the classical change of num\'{e}raire, our approach exploits an appealing representation of the prices of out-of-the money options in terms of the tail distribution functions of the underlying L\'evy process under both the original risk-neutral probability measure $\mathbb{P}$ and under the {martingale probability measure} $\mathbb{P}^{*}$ obtained when we take the stock as the num\'{e}raire; i.e. $\mathbb{P}^{*}(A):=\mathbb{E}\left( S_{t} 1_{A}\right)$ ({see e.g. Chapter 26 in \cite{Bjr09} and references therein}). The latter measure $\mathbb{P}^{*}$ is sometimes called Share measure ({see e.g} Carr\&Madan\cite{CM09}). {Our results allow us to quantify precisely the effects of a non-zero Gaussian-component in the call option prices near expiration. We find that a continuous-component volatility {of} $\sigma$ will result in an call price increase of $\frac{1}{2}\sigma^2 t^2 e^{k}\nu(k)$ {(per each dollar of the underlying spot price)}, where $\nu$ is the L\'{e}vy density and $k$ is the log-moneyness}.

We also derive the corresponding small-time asymptotic behaviour for implied volatility, {showing precisely how the implied volatility diverges to $\infty$ (see Section \ref{section:SmallTimeLevy})}. We find that the dimensionless implied variance does tend to zero as we would expect, but \textit{very} slowly; in fact slower than $t^p$ for any $p>0$, and consequently the implied volatility explodes in the small-time limit.  
Furthermore, we characterize the asymptotic behavior of the relative error of the first order approximation,  which is then used to obtain a second-order approximation for the implied volatility of out-of-the money call options. According to our numerical results (see Section \ref{Sect:NmrExmpl} for the details), the second order approximation significantly reduces the error compared to that of the first order approximation, achieving up to a two-fold relative error reduction in some cases.

 We later extend our analysis to the case of a time-changed exponential L\'evy model {$Z_{t}=X_{T_{t}}$} with an independent absolutely-continuous time change {$T_{t}=\int_{0}^{t}Y_{s}ds$} satisfying some mild moment conditions (see Section \ref{section:SmallTimeTimeChange}). The time-changed L\'evy model was proposed in \cite{CGMY03} to incorporate the volatility clustering and leverage effects {commonly} exhibited by financial price processes.  We show that  the small-time behavior of call option prices depends not only on {the triplet of the underlying L\'evy process $X$ but also} on the time-zero first and second moments of the speed process {$(Y_{t})$} and {the quantity
 \[
 	\gamma:=\lim_{t\searrow{}0}\frac{1}{t}\left[\mathbb{E} Y_{t} -\mathbb{E} Y_{0}\right],
\]
which is assumed to exist. In some sense, $\gamma$ measures the {current} average acceleration of the random clock. Under mild conditions, we show that
\[
\frac{1}{t}\,\frac{1}{S_0}\mathbb{E}(S_t-K)_{+}\,=\, \mathbb{E} Y_{0} a_{0}(k)\,+\,
 \big[ \mathbb{E}Y_{0}^{2} a_{1}(k)+\gamma a_{0}(k)\big]\,t+o(t), \quad \quad \quad \quad (t \to 0)\,,
\]
where $a_{0},a_{1}$ are the first and second order terms {appearing in the pure-L\'evy option price approximation (\ref{PureLevyApprx}). For a Cox-Ingersoll-Ross (CIR) speed process
\begin{equation*}
dY_t=\kappa(\theta-Y_t)dt+\sigma \sqrt{Y_t} dW_t,\quad Y_{0}=y_{0}
\end{equation*}
the current acceleration of the process is $\gamma=\kappa(\theta-y_0)$ and, hence, call option prices will exhibit the following small-maturity asymptotic behavior:
\[
\frac{1}{t}\,\frac{1}{S_0}\mathbb{E}(S_t-K)_{+}\,=\, y_{0} a_{0}(k)\,+\,
\big[ y_{0}^{2} a_{1}(k)+\kappa(\theta-y_0) a_{0}(k)\big]\,t+o(t), \quad \quad \quad \quad (t \to 0)\,.
\]
As seen from this expression, a mean reversion speed of $\kappa$ will increase (resp.,  decrease) the call option price when the current volatility $y_{0}$ is above (resp. below) the long-run mean volatility value $\theta$.}}

 In Section \ref{section:SmallTimeSmallMoneyness}, we also consider a small-time, small log-moneyness regime for the CGMY model of \cite{CGMY02}. The CGMY model is a particular case of the more general KoBoL class {of models}, named after the authors \cite{Kop95} (who first introduced the symmetric version of the model under the name of ``truncated L\'evy flights") and \cite{BL02}. 
 Using the fact that $\left(X_t/t^{1/Y}\right)_{t}$ converges weakly to a symmetric alpha-stable distribution with $\alpha=Y$ as $t \to 0$,
 we show that
\[
	\lim_{t\to{}0}t^{-1/Y}\frac{1}{S_0}\mathbb{E}(S_t-S_0)_{+}=\mathbb{E}^{*}(Z_{+}),
\]
for $Y\in(1,2)$,
where $Z$ is a symmetric $Y$-stable random variable under $\mathbb{P}^*$.
We then apply {this result} to small-time pricing of at-the-money call options for the {CGMY model}. Our method of proof is new and based on the following representation by Carr\&Madan\cite{CM09}
 \be\label{CMR0}
\frac{1}{S_0}\mathbb{E}(S_t-K)_{+}=\mathbb{P}^*(X_t-E> \log \frac{K}{S_0}) \,,\ee
\nind where $E$ is an independent exponential random variable under $\mathbb{P}^*$ with parameter 1.
As a corollary, we conclude that the corresponding at-the-money implied volatility $\hat{\sigma}_t(0)$ satisfies $\lim_{t \to 0}\hat{\sigma}_t(0)/t^{1/Y-\frac{1}{2}}=\sqrt{2\pi}\,\mathbb{E}^{*}(Z_{+})$.  \cite{Tnkv10} obtains a similar result in a more general model using a different approach based on a Fourier-type representation for call option prices. {Let us also remark that the method of proof introduced here can be applied to a large class of L\'evy processes whose L\'evy densities are symmetric and dominated by stable L\'evy densities, and behave like a symmetric $Y$-stable process in the small-time limit (see Remark \ref{GnrClsATM} for the details).}

 In section \ref{section:VarianceCallOptions}, we derive a similar small-time estimate for variance call options using the well known fact that the quadratic variation $[X]_t$ of a L\'{e}vy process is itself a L\'{e}vy process.   Using the main result in \cite{FL08}, we find that an out-of-the-money {variance} call option which pays $([X]_t-K)_{+}$ at time $t$ is worth the same as a European-style contract paying $(\ln\frac{S_t}{S_0})^2-K)_{+}$ at time $t$ as $t \to 0$, irrespective of the L\'{e}vy measure {$\nu(\cdot)$}.  The diffusion component of $(X_t)$ does not show up at leading order for small $t$.  See also \cite{KRMK11} for a related discussion on the difference between the small-time behaviour of variance call options on the exact quadratic variation and its discretely sampled approximation for L\'{e}vy driven models.
 
\section{Small-time asymptotics for exponential L\'{e}vy models}
\label{section:SmallTimeLevy}

Consider an exponential L\'{e}vy model for a stock price process \be
\label{eq:FVLevyModel} S_t=S_0 e^{X_t} \ee \nind where $(X_t)$ is a
L\'{e}vy process defined on a complete probability space $(\Omega, \mathbb{P},\mathcal{F})$ with generating triplet
$(\sigma^2,b,\nu)$.  {We are assuming zero interest rate and dividend yield for simplicity \footnote{{For a non-zero constant interest rate $r$ and dividend rate $q$, the results in this paper will not be qualitatively any different, because we can just replace the stock price process $(S_{t})$ with the forward price process $(e^{-(r-q)t}S_{t})_{t}$, which is a martingale (see, e.g. {Chapter 11 in \cite{CT04}}).}} and that $\mathbb{P}$ represents a risk-neutral pricing measure}.  We assume that $\int_{|x|>1}e^{x}\nu(dx)<\infty$ and that the following condition is satisfied
\be \label{eq:Martingale}
b+\frac{1}{2}\sigma^2+ \int_{-\infty}^{\infty} (e^x-1- x1_{|x|\le
1})\nu(dx)=0, \ee \nind so that $S_t=S_0 e^{X_t}$ is {indeed a $\mathbb{P}$}-martingale relative to its own filtration.

Throughout the paper, we also assume that the L\'evy measure $\nu(dx)$ admits a \emph{positive} density, denoted $\nu(x)$, and that this density is $C^{1}$ in $\mathbb{R}\backslash\{0\}$ satisfying  $\sup_{|x|>\varepsilon}\nu(x)<\infty$ for every $\varepsilon>0$. The choice between the L\'evy measure $\nu(dx)$ and density $\nu(x)$ should be clear from the context. Under the previous standing condition, Figueroa-L\'{o}pez\&Houdr\'{e} \cite{FLH09} show the following result (see Remark 3.3 and Proposition 3.4 therein):
\begin{thm}\label{ThFLH09}\cite{FLH09}.
Let $y>0$ be fixed. Then, we have the following small-time behaviour for the distribution function of $X_t$:
\label{thm:JoseThm}
$$
\frac{1}{t}\,\mathbb{P}(X_t \ge y)= \nu[y,\infty) +\frac{1}{2}t {d_2(y)}
+o(t) \quad \quad \quad \quad (t \to 0)\,,
$$
\nind where
\bq
d_2(y)&=&d_2(y;b,\sigma,{\nu})=-\sigma^2 \nu'(y)+ 2b \nu(y)- \nu[y,\infty)^2+\nu(\frac{1}{2}y,y)^2 \nn \\ &+& 2\int_{-\infty}^{-\frac{1}{2}y} \int_{y-x}^y \nu(u) \nu(x) du dx - 2\nu(y)\int_{\frac{1}{2}y <|x|<1} x\nu(x)  dx +  2\nu(y)\int_{-\frac{1}{2}y}^{\frac{1}{2}y}\int_{y-x}^y (\nu(u)-\nu(y))\nu(x)du dx\,.\label{ExprCoefd2}
\eq
\nind
Furthermore, if the pure-jump component of $(X_t)$ has finite variation, i.e. $\int_{|x|\le 1} |x|\nu(dx)<\infty$, then $d_2$ simplifies to
\begin{equation}\label{ScndCoefBndVar}
d_2(y)=d_2(y;b,\sigma,\nu)=-\sigma^2 \nu'(y)+ 2b_0 \nu(y)- \nu[y,\infty)^2+ \int_0^y \int_{y-x}^y \nu(u) \nu(x) du dx - 2\int_y^{\infty} \int_{-\infty}^{y-x} \nu(u)  \nu(x) du dx\,,
\end{equation}
\nind where $b_0$ is the drift of the {pure-jump component of $(X_{t})$} defined by $b_{0}=b-\int_{|x|\le 1} x \nu(dx)$.
\end{thm}

\begin{rem}
{
	The double integrals in (\ref{ExprCoefd2}) and (\ref{ScndCoefBndVar}) are well-defined. For instance, by the symmetry of $s(u)s(x)$ about the line $u=x$,
	\begin{align}
		\int_0^y \int_{y-x}^y \nu(u) \nu(x) du dx&=\int_0^{y/2} \int_{y-x}^y \nu(u)\nu(x) du dx+\int_{y/2}^{y} \int_{y-x}^{y/2} \nu(u)\nu(x) du dx+\int_{y/2}^{y} \int_{y/2}^{y} \nu(u)\nu(x) du dx \nn
		\\
		&=2\int_0^{y/2} \int_{y-x}^y \nu(u)\nu(x) du dx+\int_{y/2}^{y} \int_{y/2}^{y} \nu(u)\nu(x) du dx\label{HTFTH}
		\\
		&\leq 2[\sup_{u\in(y/2,y)}\nu(u)]\int_0^{y/2} x \nu(x)dx+[(y/2)
		\sup_{u\in(y/2,y)}\nu(u)]^{2},\nn
	\end{align}
	which is finite because $\int_{|x|\leq{}1}|x|\nu(x)dx<\infty$ (being $X$ a bounded variation process).  To obtain the bound on the first integral (\ref{HTFTH}), we  used the fact that the range for $x$ is from $0$ to $y/2$ and, hence, the maximal range for u in the inner integral is $u\in[y/2,y]$.
Similarly, by Fubini's theorem, 
	\begin{align*}
		\int_y^\infty \int_{-\infty}^{y-x} \nu(u) \nu(x) du dx&=\int_{-\infty}^0 \int_{y}^{y-u} \nu(u) \nu(x) dx du\\
		&\leq \int_{-1}^0 \int_{y}^{y-u} \nu(x) dx  \nu(u)du+\int_{-\infty}^{-1}\int_{y}^{y-u} \nu(x) dx  \nu(u)du\\
		&\leq [\sup_{x>y}\nu(x)]\int_{-1}^{0} (-u) \nu(u)du+
		\int_{-\infty}^{-1} \nu(u)du\int_{y}^{\infty} \nu(x) dx<\infty.
	\end{align*}
	}
\end{rem}

In the following proposition, we use Theorem \ref{thm:JoseThm} to establish a small-time estimate for the price of an out-of-the-money call option under the model in \eqref{eq:FVLevyModel}:
\begin{prop}
\label{prop:SmallTimeCallsLevy}
Assume that \be\label{eq:ExpCondition}(i)\;\int_{|x|>1} e^x \nu(x)dx<\infty\,
\quad \text{and}\quad (ii)\;\sup_{|x|>\varepsilon} e^{x}\nu(x)<\infty,\ee \nind
for any $\varepsilon>0$. Then we have the following small-time expansion for the price of a call option with strike $K>S_0$
\be
\label{eq:CallOptionAsymptotic}
\frac{1}{t}\,\mathbb{E}(S_t-K)_{+}\,=\,S_0 \int_{-\infty}^{\infty} (e^x-e^k)_{+}\nu(x)dx \,+\, \frac{1}{2}S_0 \big[d_2^*(k)- e^k d_2(k)\big]t+o(t) \quad \quad \quad \quad (t \to 0)\,,
\ee
\nind
where $k=\log \frac{K}{S_0}>0$ is the log-moneyness and $d_2^*(k)=d_2(k;b^{*},\sigma,\nu^*)$ with $b^*$ and ${\nu^*}$ given by
\begin{equation}\label{NLevyTriplet}
	{\nu^{*}(x)=e^{x}\nu(x)}\quad\text{and}\quad
	b^{*}=b+\int_{|x|\leq{}1} x \left(e^{x}-1\right)\nu(x)dx +\sigma^{2}.
\end{equation}
\end{prop}
\begin{rem}
This result sharpens the asymptotic behavior (\ref{ARFO}), {established by Levendorskii \cite{Lev08} for a class of multi-factor L\'{e}vy and Markov models under certain technical conditions. As explained in the introduction, for a L\'evy model, these conditions were relaxed by Roper\cite{Rop10} and Tankov \cite{Tnkv10}.}
Note also that, by imposing that $\nu$ has a positive L\'{e}vy density, we are precluding the Black-Scholes case where there is a non-zero diffusion component with volatility $\sigma$ and \textit{zero} jump component, for which the implied volatility is just constant and equal to $\sigma$.
\end{rem}
\begin{proof}
Without loss of generality, we assume that $(X_{t})$ is the canonical process $X_{t}(\omega)=\omega(t)$ defined on $\Omega=\mathbb{D}([0,\infty),\mathbb{R})$ (the space of right-continuous functions with left limit $\omega:[0,\infty)\to\mathbb{R}$) and equipped with the $\sigma-$field $\mathcal{F}=\sigma(X_{s}:s\geq{}0)$ and the right-continuous filtration $\mathcal{F}_{t}:=\cap_{s>t}\sigma(X_{u}:u\leq{}s)$.
{Following the density transformation construction of Sato\cite{Sat99} (see Definition 33.4 and Example 33.4 therein) and using the martingale condition \eqref{eq:Martingale}, we define $\mathbb{P}^{*}$ {on $(\Omega,\mathcal{F})$} such that
\begin{equation}\label{DSM}
	\mathbb{P}^{*}(B)=\mathbb{E}\left(e^{X_{t}} 1_{B} \right),
\end{equation}
for any $t>0$ and $B\in\mathcal{F}_{t}$.  As explained in the introduction, {we can interpret $\mathbb{P^*}$ as} the martingale measure associated with using the stock price as the num\'{e}raire}.

Let us first note that the price of a call option {can be decomposed as follows}:
\bq
\mathbb{E}(S_t-K)_{+}&=& \mathbb{E}(S_t 1_{S_t \ge K})-K \mathbb{P}(S_t \ge K)\nn \\
&=& S_{0}\mathbb{E}(e^{X_t} 1_{S_t \ge K})-S_0 e^k \mathbb{P}(X_t \ge k) \nn \\
&=& S_0\mathbb{P}^*(X_t \ge k)-S_0 e^k \mathbb{P}(X_t \ge  k)
\label{eq:LevyCallDecomposition}
\eq
One can check that $(X_{t})$ is a L\'evy process under $\mathbb{P}^{*}$ with characteristic triplet $(b^{*},\sigma^{2},\nu^{*})$.
For this result see the more general Theorem 33.1 in Sato\cite{Sat99}.  Finally,  applying Theorem \ref{thm:JoseThm} to the probabilities under $\mathbb{P}$ and $\mathbb{P}^*$ in \eqref{eq:LevyCallDecomposition}, we have
\be
{\frac{1}{t}\,\mathbb{E}(S_t-K)_{+}\,=\,S_0 \int_{k}^{\infty} e^x\nu(x)dx -K  \int_{k}^{\infty}\nu(x)dx  \,+\, \frac{1}{2}S_0 d_2^*(k) t- \frac{1}{2}K d_2(k)t+o(t) \quad \quad \quad \quad (t \to 0)\,,}
\ee
\nind which simplifies to \eqref{eq:CallOptionAsymptotic}.
\end{proof}

{
\begin{rem}\label{RmrDrftVolEffct}
	Let us note that for a bounded variation process, the drift of $X$ under the Share measure $\mathbb{P}^*$ is the same as the drift under the measure $\mathbb{P}$. Indeed, denoting $b^{*}_{0}$ the drift under  $\mathbb{P}^*$, we have that
	\[
		b^{*}_{0}=b^*-\int_{\{|x|\leq{}1\}} x\nu^*(x)dx=b+\int_{|x|\leq{}1} x \left(e^{x}-1\right)\nu(x)dx
		-\int_{\{|x|\leq{}1\}} xe^{x}\nu(x)dx=b_{0}.
	\]
	Also, note that the call price approximation (\ref{eq:CallOptionAsymptotic}) is independent of $b$. Indeed, let
	\[
		R(y;\nu):=d_2(k;b,\sigma,\nu)-\left(-\sigma^{2}\nu'(y)+2 b \nu(y)\right),
	\]	
	which depends only on $\nu$ as seen from the expression of $d_{2}$ in (\ref{ExprCoefd2}).
	Then, using (\ref{NLevyTriplet}), the second order term in (\ref{eq:CallOptionAsymptotic}) can be simplified as follows
	\begin{align*}
		a_1(k):= a_1(k;b,\sigma,\nu)&:=\frac{1}{2}\big[d_2(k;b^{*},\sigma,\nu^*)-e^{k} d_2(k;b,\sigma,\nu)\big]\\
		&=\frac{\sigma^{2}}{2}e^{k}\nu(k)
		+  e^{k}\nu(k)\int_{|x|\leq{}1} x \left(e^{x}-1\right)\nu(x)dx+
		\frac{1}{2}\left[R(k;\nu^{*})-e^{k}R(k;\nu)\right],
	\end{align*}
which does not depend on $b$. The previous expression also shows that
	\[
		a_1(k;b,\sigma,\nu)-a_1(k;b,0,\nu)=\frac{\sigma^{2}}{2} e^{k}\nu(k),
	\]
and, hence, a non-zero volatility of $\sigma$ has the effect of increasing the call price approximation by $\frac{\sigma^{2}	t^{2}}{2} e^{k}\nu(k)$.
\end{rem}
}

\subsection{Implied volatility}\label{SSec:ImplVol}

  Let $\hat{\sigma}_t(k)$ denote the Black-Scholes implied volatility at log-moneyness $k$ {and} maturity $t$ with zero interest rates, and let $V(t,k)=\hat{\sigma}_t(k)^2 t$ denote the dimensionless \textit{implied variance}.  Let
 \bq\label{TermFirstSecondTerms}
  a_0(k):= \int_{-\infty}^{\infty} (e^x-e^k)_{+}{\nu(dx)}
  \quad \text{and}\quad
   a_1(k):= \frac{1}{2}\big[d_2^*(k)- e^k d_2(k)\big]
\eq
denote the (normalized) leading order and correction terms in \eqref{eq:CallOptionAsymptotic}.  By put-call parity, the dominated convergence theorem, {and the stochastic continuity of the L\'evy process $(X_{t})$}, we have
$$
\lim_{t \to 0}\mathbb{E}(S_t-K)_{+} = (S_0-K)_{+}\,,
$$
\nind and from this we can show that $V(t,k) \to 0$ as $t \to 0$.  The following corollary shows more precisely how $V(t,k)\to 0$ as $t \to 0$ {and, hence,} sharpening a result in Tankov \cite{Tnkv10} (Proposition 4 therein):
\begin{thm}
\label{thm:LevyImpliedVariance}
For the exponential L\'{e}vy model in \eqref{eq:FVLevyModel}, we have the following small-time behavior for the implied variance $V(t,k)$ for $k>0$
\be
\label{eq:ImpliedVarianceLevy}
V(t,k)= V_0(t,k) \big[1+
V_1(t,k) \,+\, o(\frac{1}{\log\frac{1}{t}}) \big]
 \quad \quad \quad \quad (t \to 0),
\ee
\nind
where
\bq
V_0(t,k)&=& \frac{\frac{1}{2}k^2}{\log (\frac{1}{t})}\,, \nn \\
V_1(t,k)&=&  \frac{1}{\log(\frac{1}{t})}\log\left[\frac{4 \sqrt{\pi}a_{0}(k)e^{-k/2}}{k}\left[\log\left(\frac{1}{t}\right)\right]^{3/2}\right]
\,.\eq
\end{thm}
\begin{proof}
See Appendix \ref{section:Proofs}.
\end{proof}
\begin{rem}
Multiplying \eqref{eq:ImpliedVarianceLevy} by $1/t$, we have the following expansion for the implied volatility
\be
\label{eq:ImpliedVolatilityLevy}
\hat{\sigma}^2_t(k)= \frac{\frac{1}{2}k^2}{t\log(\frac{1}{t})}\big[1 \,+\,
V_1(t,k) \,+\, o(\frac{1}{\log\frac{1}{t}})]\,
 \quad \quad \quad \quad (t \to 0),
\ee
\nind and we see that $\hat{\sigma}^2_t(k)\to \infty$
as {$t \to{}0^{+}$}, as is well documented in e.g. Carr\&Wu\cite{CW03}
{(see also Roper\cite{Rop10} and \cite{Tnkv10})}.  The leading order term agrees with that obtained in Tankov\cite{Tnkv10} and, moreover, we see that
$$
 [t\log(\frac{1}{t})]^{\frac{1}{2}}\hat{\sigma}_t(k) \sim |k|/\sqrt{2},\quad \quad \quad \quad (t \to 0)\,,
$$
\nind so the (re-scaled) leading order implied volatility smile is V-shaped and independent of $\nu$, except that we require $\nu$ to be non-zero.
\end{rem}
\begin{rem}  $V(t,k)=O(\frac{1}{\log\frac{1}{t}})$, so $V(t,k)\to 0$ but slowly; in fact slower than $t^p$ for any $p>0$. In particular, for a given desired ``precision" bound $\varepsilon\ll 1$, we will need $t=O(e^{-1/\epsilon})$ to ensure that $V(t,k) =O(\epsilon) $ and for the $\frac{1}{\log\frac{1}{t}}$ error term in \eqref{eq:ImpliedVarianceLevy} to be $O(\epsilon)$.  For this reason, the call option estimate \eqref{eq:CallOptionAsymptotic} is more useful than the implied volatility estimate  \eqref{eq:ImpliedVolatilityLevy} in practice. {We remark that in Corollary 8.3 of the very recent article by Gao\&Lee\cite{GL11}, the authors give an expansion which sharpens \eqref{eq:ImpliedVarianceLevy}, but proving their result is more involved and
requires several preliminary lemmas}
\end{rem}
\begin{rem}\label{RmrkBlowing}
	Based on high-frequency statistical methods for It\^{o} semimartingales, several empirical studies have statistically rejected the null hypothesis of either a purely-jump or a purely-continuous model (see, e.g., \cite{AJ09b}, \cite{AJ10}, \cite{BNS06}).  If this {really is} the case, then our results show that theoretically, the small-maturity smile must tend to infinity, if put/call options are priced correctly.  Nevertheless, this effect is often obscured in reality by market practicalities - high bid/offer spreads, daycount/settlement conventions, and times when the market is closed.  However, even if we cannot trade an option with infinitesimally small maturity in practice, we can still look at rate at which the implied volatility smile steepens as the maturity goes small; typically it is difficult to fit the one of the fashionable class of purely continuous models (e.g. Heston, SABR, and other local-stochastic volatility hybrid models) to this kind of data, with realistic parameters.  Carr\&Wu\cite{CW03}'s study of S\&P 500 \textit{option price} data (in contrast to the previous \textit{statistical} approaches) also suggests that the sample path of the index contains both continuous and discontinuous martingale components (working under a risk neutral measure), and that, while the presence of the jump component varies strongly over time, the continuous component is omnipresent. \\
\indent {In the same vein,} A\"it-Sahalia\&Jacod\cite{AJ09a} define a \textit{jump activity index} to test for the presence of jumps, which for a L\'{e}vy process coincides with the Blumenthal-Getoor index of the process.  \cite{AJ09a} also {proposes} estimators of this index for a discretely sampled process and derive the estimators' properties.
These estimators are applicable despite the presence of a Brownian component
in the process, which makes it more challenging to infer the characteristics
of the small, infinite activity jumps. When the method is applied to high frequency
stock returns, \cite{AJ09a} found evidence of infinitely active jumps in the data
and they were able to estimate the index of activity.
\end{rem}

\section{Time-changed L\'{e}vy processes}
\label{section:SmallTimeTimeChange}
\subsection{A formula for out-of-the-money call option prices}
In addition to the L\'evy process $(X_{t})$ of Section \ref{section:SmallTimeLevy}, we now consider a random clock $(T_{t})$ {defined on $(\Omega,\mathbb{P},\mathcal{F})$ and independent of $X$}. A random clock is a right-continuous non-decreasing process such that $T_{0}=0$. We consider a time-changed L\'evy model of the form
\begin{equation}\label{TCM}
	S_{t}:=S_{0}e^{Z_{t}}, \quad {\text{with}\quad  Z_{t}:=X_{T_{t}}}.
\end{equation}
As explained in the introduction, this type of model is important because it can incorporate volatility clustering effects.

Given that $e^{X_{t}}$ is a martingale under $\mathbb{P}$ (relative to the natural filtration generated by $X$), it is known that $(S_{t})$ above is a martingale under $\mathbb{P}$ relative to the natural filtration generated by the random clock $T_{t}$ and the time-changed process $Z_{t}$ (see Lemma 15.2 in \cite{CT04}).  Note also that our simplifying assumption (\ref{eq:Martingale}) implies that
\be \label{TCMFSA}
S_t=S_0\frac{e^{Z_t}}{\mathbb{E}(e^{Z_t})}
\ee
because $\mathbb{E}(e^{Z_t})=1$.  {\cite{CGMY03} (Section 4.2) shows that the price process (\ref{TCMFSA}) is free of static arbitrage opportunities. Furthermore,}  under certain conditions (e.g. if $X$ has infinite jump activity and $(T_t)$ is continuous), $\sigma(T_{u}:u\leq{}t)\subset \sigma(X_{T_{u}}:u\leq{}t)$ (see, e.g., Theorem 1 in \cite{Win01}), and hence \eqref{TCM} will be a martingale relative to the filtration generated by only the time-changed process $(Z_{t})$ or, equivalently, the filtration generated by the stock-price process $(S_{t})$. In that case, the model (\ref{TCM}) will be free of dynamic arbitrage opportunities by {the sufficiency part of {the} First  Fundamental Theorem of Asset Pricing}.

Let $\mathcal{N}$ be the set of $\mathbb{P}$-null sets of $\mathcal{F}$ and define a probability measure $\widetilde{\mathbb{P}}$ on $\widetilde{\mathcal{F}}:=\sigma(Z_{t}, T_{t}:t>{}0)\vee \mathcal{N}$ such that, for any $t>0$,
\begin{equation}\label{DSM}
	\widetilde{\mathbb{P}}(B)=\mathbb{E}\left(e^{Z_t} 1_{B} \right),
\end{equation}
whenever  $B\in\widetilde{\mathcal{F}}_{t}:=\sigma(Z_{u},T_{u}:u\leq t)\vee \mathcal{N}$.  We note that $\widetilde{\mathbb{P}}$ is well defined since $\{e^{Z_{t}}\}_{t\geq{}0}$ is a $\mathbb{P}$-martingale relative to $\{\widetilde{\mathcal{F}}_{t}\}_{t\geq{}0}$.
The following proposition will play a key role in the sequel:
\begin{prop}
\label{prop:OptionPriceReprTimeChangedLevy}
Suppose that the assumptions of Proposition \ref{prop:SmallTimeCallsLevy} are satisfied and let $(b^{*},\sigma^{2},\nu^{*})$ be defined as in (\ref{NLevyTriplet}). Then, under $\widetilde{\mathbb{P}}$, the process $\left(Z_t\right)$ in (\ref{TCM}) has the same distribution as a L\'{e}vy process with the characteristic triplet $(b^{*},\sigma^{2},\nu^{*})$  evaluated at the independent random clock $T_t$.
\end{prop}
\begin{proof}
 	Fix $0=t_{0}<\dots<t_{n}= t<\infty$ and $u_{1},\dots,u_{n}\in\mathbb{R}$. Then, using the independence between $T$ and $X$,
\begin{align*}
	\widetilde{\mathbb{E}}(\exp\{i\sum_{j=1}^{n}u_{j} (Z_{t_{j}}-Z_{t_{j-1}})\})&=\mathbb{E}(\exp\{Z_{t}+i\sum_{j=1}^{n}u_{j} (Z_{t_{j}}-Z_{t_{j-1}})\})
=\mathbb{E}(\exp\{\sum_{j=1}^{n}i(u_{j}-i) (X_{T_{t_{j}}}-X_{T_{t_{j-1}}})\})\\
	&=\mathbb{E}(\exp\{\sum_{j=1}^{n} (T_{t_{j}}-T_{t_{j-1}})\psi(u_{j}-i)\})=\mathbb{E}(\exp\{\sum_{j=1}^{n} (T_{t_{j}}-T_{t_{j-1}})\psi^{*}(u_{j})\})\,.
\end{align*}
The last expression corresponds to the characteristic function of a process of the form $X^{*}_{T_{t}}$, where $(X^{*}_{t})$ is a L\'evy process with triplet $(b^{*},\sigma^{2},\nu^{*})$ defined on $(\Omega,\mathbb{P},\mathcal{F})$ and \emph{independent} of the random clock $(T_{t})$.
\end{proof}

In light of the previous result, we have the following representation for call option prices:
\bq
\mathbb{E}(S_t-K)_{+}&=& \mathbb{E}(S_t 1_{S_t \ge K})-K \mathbb{P}(S_t \ge K)\nn \\
&=& S_0\mathbb{E}(\exp(Z_{t}) 1_{Z_{t} \ge k})-S_{0} e^{k} \mathbb{P}(S_t \ge K).\nn\\
&=& S_0 \widetilde{\mathbb{P}} \left(Z_t \ge k\right)-S_0 e^k \mathbb{P}(Z_t \ge  k).
\label{CPTCL}
\eq
We emphasize again that, under $\widetilde{\mathbb{P}}$, $Z_{t}$ has the same distribution as a L\'evy process with characteristic {triplet} $(b^{*},\sigma^{2},\nu^{*})$ evaluated at an independent random clock $(T_{t})$. Hence, as for the pure-L\'evy model case, the problem of finding small-time expansions for out-the-money option prices reduces to finding small-time asymptotics of the corresponding distribution functions.
\subsection{Small-time asymptotics for the time-changed L\'evy model}
In this section, we determine the asymptotic behavior of out-the-money call option prices.
We consider random clocks $(T_{t})$ that are absolutely continuous with non-negative rate process $(Y_{t})$ (i.e. $T_t=\int_0^t Y_s ds$) such that $Y_{0}>0$. {We will also refer to the following conditions in the sequel}:
\begin{align}\label{AN3}
	&{\bf (i)}\;\; \mathbb{E} Y_{t} -\mathbb{E} Y_{0} =O(t), \quad
	{\bf (ii)}\;\; \limsup_{t\searrow{}0}\, \mathbb{E} Y^{2}_{t}<\infty,\quad
	{{\bf (iii)}\;\; \lim_{t\searrow{}0}\frac{1}{t}\left[\mathbb{E} Y_{t} -\mathbb{E} Y_{0}\right]=\gamma\in[0,\infty)},\\
	\label{AN4}
	&{\bf (iv)}\;\; \limsup_{t\searrow{}0}\, \mathbb{E} Y_{t}^{3}<\infty, \quad
 	{\bf (v)}\;\; \lim_{t\searrow{}0} \frac{1}{t^{2}}\mathbb{E} T_{t}^{2}=\rho\in(0,\infty).
\end{align}
In the case that $(Y_{t})$ is a stationary process with finite moment of third order, $\mathbb{E}Y_{t}^{k}$ is constant for $k=1,\dots,3$ and (i)-{(iv)} are automatically satisfied. Also, if $Y_{t}\to{}Y_{0}$ and {(iv)} are satisfied, then {(v)} holds true with $\rho=\mathbb{E}\, Y_{0}^{2}$. Indeed, note first that $\lim_{s\to{}0}\mathbb{E}\,Y_{s}^{2}=\mathbb{E}\,Y_{0}^{2}$ since ${(Y_{t}^{2})_{t<t_{0}}}$ are uniformly integrable for small enough {$t_{0}$} by {(iv)} above. Also, since
\(
{	T_{t}^{2}/t^{2}\leq \int_{0}^{t}Y_{s}^{2}ds/t}
\)
(by Jensen's inequality) and $\lim_{t\to{}0}\mathbb{E} \int_{0}^{t}Y_{s}^{2}ds/t=\mathbb{E} \lim_{t\to{}0}\int_{0}^{t}Y_{s}^{2}ds/t$, so the dominated convergence theorem implies that
\[
	\lim_{t\to{}0}\frac{1}{t^{2}} \mathbb{E} T_{t}^{2}=
	\mathbb{E} \lim_{t\to{}0}\left(\frac{1}{t}\int_{0}^{t}Y_{s}ds\right)^{2}
	=\mathbb{E} Y_{0}^{2}.
\]
The following result gives the small-time asymptotic behavior of the tail distributions of time-changed L\'evy models:
\begin{thm}\label{TAB2}
  	Suppose that the conditions of Theorem \ref{ThFLH09} are satisfied as well as conditions \emph{(i)-(ii)} of (\ref{AN3}). Then,
\begin{equation}\label{ABTa}
			\mathbb{P}(Z_{t}\geq{} x)= t\mathbb{E}Y_{0} \nu[x,\infty) \left[1+O(t)\right],\quad \quad \quad \quad (t \to 0).
\end{equation}
If, additionally, conditions {\emph{(iii)-(v)}} of (\ref{AN4}) are satisfied, then
\begin{equation}\label{ABTb}
			\mathbb{P}(Z_{t}\geq{} x)= t \mathbb{E}Y_{0} \nu[x,\infty) +   {\frac{1}{2}\left(\rho \,d_{2}(x)+\gamma\nu[x,\infty)\right)} t^{2}+ o(t^{2}),\quad \quad \quad \quad (t \to 0),
		\end{equation}
	where $d_{2}$ is the same as in Theorem \ref{ThFLH09}.		
\end{thm}
\begin{proof}
See Appendix \ref{section:Proofs}.
\end{proof}
\begin{rem}
	A very popular rate process in applications is the Cox-Ingersoll-Ross (CIR) diffusion process, defined by
\begin{equation}\label{CIRM}
dY_t=\kappa(\theta-Y_t)dt+\sigma \sqrt{Y_t} dW_t,
\end{equation}
where $(W_t)$ is a standard Brownian motion, $Y_0$ is an integrable positive random variable independent of $W$, and
$\kappa,\theta,\sigma>0$ are such that $\kappa \theta/\sigma^{2}>1/2$ (which ensures that $Y=0$ is an inaccessible boundary). If $Y_{0}\sim \Gamma(\frac{2\theta \kappa}{\sigma^{2}},\frac{\sigma^{2}}{2\kappa})$, the proces $(Y_{t})$ is stationary and  $\mathbb{E} Y_{t}^{k}$ is finite and constant in $t$ for any $k\geq{}1$. In particular, (i)-{(v)} are satisfied with $\rho=\mathbb{E}Y_{0}^{2}$. In the non-stationary case, it is known that $\mathbb{E}Y_{t}-\mathbb{E}Y_{0}=\left(\theta-\mathbb{E}Y_{0}\right)\left(1-e^{-\kappa t}\right)$ and {(i) \& (iii) are satisfied with $\gamma={\kappa}(\theta-\mathbb{E}Y_{0})$. The other conditions in (\ref{AN3}-\ref{AN4}) will also hold true. Thus we conclude that the time-changed L\'evy model with CIR speed process satisfies:
\[
	\mathbb{P}(Z_{t}\geq{} x)= t \mathbb{E}Y_{0} \nu[x,\infty) +   \left(\mathbb{E} Y_{0}^{2} \,d_{2}(x)+{\kappa}(\theta-\mathbb{E}Y_{0})\nu[x,\infty)\right) \frac{1}{2} t^{2}+ o(t^{2}),\quad \quad \quad \quad (t \to 0).	
\]}
\end{rem}
We are now ready to give the small-time asymptotic behavior of out-the-money call option prices and the corresponding implied volatility:
\begin{cor}
Under the conditions of Proposition \ref{TAB2}, we have the following small-time expansions
\be
\label{eq:CallOptionAsymptoticTCM}
\frac{1}{t}\,\mathbb{E}(S_t-K)_{+}\,=\,S_0 \mathbb{E} Y_{0} a_{0}(k)\,+\,
S_0 {\big[ \rho a_{1}(k)+\gamma a_{0}(k)\big]}t+o(t) \quad \quad \quad \quad (t \to 0)\,,
\ee
\nind
where $k=\log K/S_0>0$ and {$a_{0},a_{1}$ are the first and second order terms of the call price approximation (\ref{eq:CallOptionAsymptotic}) as defined in (\ref{TermFirstSecondTerms})}.
Furthermore, we have the following small-time behaviour for the implied variance $V(t,k)$ for $k>0$
\be
\label{eq:ImpliedVarianceTCLevy}
V(t,k)= V_0(t,k) \big[1+
V_1(t,k) \,+\, o(\frac{1}{\log\frac{1}{t}}) \big]
 \quad \quad \quad \quad (t \to 0),
\ee
\nind
where
\bq
V_0(t,k)&=& \frac{\frac{1}{2}k^2}{\log (\frac{1}{t})}\,, \nn \\
V_1(t,k)&=&  \frac{1}{\log(\frac{1}{t})}\log\left(\frac{4 \sqrt{\pi}\,\mathbb{E}(Y_0)a_{0}(k)e^{-k/2}}{k}\left[\log\left(\frac{1}{t}\right)\right]^{3/2}\right)
\,.\eq
\end{cor}
\begin{proof}
	The expansion (\ref{eq:CallOptionAsymptoticTCM}) follows from the representation (\ref{CPTCL}) and (\ref{ABTb}). The asymptotics (\ref{eq:ImpliedVarianceTCLevy}) follows from the proof of Thorem \ref{thm:LevyImpliedVariance}.
\end{proof}

{
\begin{rem}
	As it was indicated before, the time-changed L\'evy model (\ref{TCM}) was introduced to account for the volatility clustering exhibited by financial time series. Indeed, the process $(Y_{t})_{t}$ controls the speed of the random clock so that when $Y_{t}$ is high, the random clock runs faster and, hence, the price process exhibits more variability. Another approach to incorporate stochastic volatility is via stochastic integration along the lines of the following jump-diffusion model
	\begin{equation}\label{SVMGH}
		d \ln(S_{t}/S_{0})=\mu(Y_{t}) dt +\sigma(Y_{t}) d W_{t}^{1}+d Z_{t}, \quad d Y_{t}=\alpha(Y_{t})dt +\gamma(Y_{t}) d W_{t}^{(2)},
	\end{equation}
	where $W^{(1)}$ and $W^{(2)}$ are {two (possibly correlated)} Brownian motions and $Z$ is a pure-jump process. For {a} comparison of these two method{s}, we refer the reader to Chapter 15 of \cite{CT04}. Recently, \cite{FLGH11} have provided small-time expansions for vanilla option prices under the stochastic model (\ref{SVMGH}) when $Z$ is a {pure-jump} L\'evy process independent of $Y$.
\end{rem}
}

\section{Small-time, small log-moneyness asymptotics}\label{section:SmallTimeSmallMoneyness}

In this section, we survey the behavior of $\mathbb{P}(X_t \ge k)$ for a L\'{e}vy process $X$, when $t\to{}0$ and $k=k_{t}$ also converges to zero at an appropriate rate.  We can think of this scaling as a \textit{small-time, small log-moneyness} regime. As an application, we deduce the asymptotic behavior of at-the-money call option prices for a {CGMY model}.

\subsection{L\'evy models with non-zero Brownian component}

Several financial models in the literature consist of a L\'evy model with non-zero Brownian component. The most popular models of this kind are the Merton model and Kou model determined by the characteristic functions
\begin{align*}
	\mathbb{E}(\exp(iu X_t))&=\exp[ t(ib u-\frac{1}{2}\sigma^2 u^{2} +  iu \lambda\,(\frac{p}{\lambda_{+}-iu}-\frac{1-p}{\lambda_{-}+iu}))],\\
	\mathbb{E}(\exp(iu X_t))&=\exp[ t(ib u-\frac{1}{2}\sigma^2 u^{2} +  \lambda\,(e^{-\delta^{2}u^{2}/2+i\mu u}-1))].
\end{align*}
It turns out that, for a general L\'evy process $(X_{t})$ with $\sigma\neq{}0$,
$$
\lim_{t \to 0}\mathbb{E}(\exp(iu X_t/\sqrt{t}))= \exp(-\frac{1}{2}\sigma^2 u^{2})\,,
$$
(see e.g. pp. 40 in \cite{Sat99} for a formal proof).  The right-hand side is the characteristic function of a Normal $N(0,\sigma^2)$ random variable $Z$, thus $(X_t/\sqrt{t})$ converges weakly to a Normal distribution with variance $\sigma^2$ and
$$
\lim_{t \to 0} \mathbb{P}(X_t/\sqrt{t} > x) =\mathbb{P}(Z > x)\,.
$$

\subsection{The CGMY model and other tempered stable models}

The {so-called} CGMY model is a pure-jump L\'evy process determined by a L\'{e}vy density of the form
 \bq
    \label{eq:HestonSmallT}
	{\nu(x)=\frac{C e^{-Mx}}{x^{1+Y}}\, {\bf 1}_{\{x> 0\}}+\frac{C e^{Gx}}{|x|^{1+Y}}\,{\bf 1}_{\{x<0\}}}.
\eq
for $C,G,M>0$ and $Y\in(0,2)$.   {As explained in the introduction, the CGMY model is a particular case of the more general KoBoL class {of models}, named after the authors \cite{Kop95} (who first introduced the symmetric version of the model under the name of ``truncated L\'evy flights") and \cite{BL02}. The term CGMY was introduced later on by Carr et al. \cite{CGMY02}. This process} is a tempered stable process (see Section 4.5 in Cont\&Tankov\cite{CT04}), and its characteristic function is given as
\begin{equation}\label{CFCGMY}
	{\phi_t(u)=\mathbb{E}(e^{iu X_t})=\exp\left[t \,C \Gamma(-Y)\left\{(M-iu)^Y+(G+iu)^Y-M^Y-G^Y\right\}+i {\hat{b}}ut\right]},
\end{equation}
\nind for $Y\neq{}1$ {and some constant $\hat{b}\in\mathbb{R}$} (see \cite{CT04} for the formula when $Y=1$). We note that we must have $M>1$ for (\ref{eq:ExpCondition}) to be satisfied, and under this condition, $X$ is again a CGMY process under $\mathbb{P}^{*}$ with parameters $C^{*}=C$,  $Y^{*}=Y$,
$M^{*}=M-1$, and $G^{*}=G+1$. {In the bounded variation case ($Y<1$), $\hat{b}$ coincides with the drift $b_{0}$.}

The following result characterizes the small-time behavior of $\mathbb{P}(X_t > k_{t})$ with small log-moneyness  $k_{t}\sim x t^{1/Y}$.
\begin{prop}\label{prop:CGMYConv}
For the CGMY model with $Y\in(1,2)$, $(X_t/t^{1/Y})$ converges weakly to a symmetric {$Y$-stable} distribution as $t\to 0$. Concretely,
$$
\lim_{t \to 0} \mathbb{P}(X_t/t^{1/Y} > x) =\mathbb{P}(Z > x)\,,
$$
\nind where $Z$ is a symmetric {$Y$-stable random variable with scale parameter $c=(2C\Gamma(-Y)|\cos(\frac{1}{2}Y\pi)|)^{1/Y}$; i.e. $Z$ has} characteristic function
$$
	{{\zeta}(u)=\exp(-2C\Gamma(-Y)|\cos(\frac{1}{2}Y\pi)|\,|u|^{Y})}\,.
$$
\end{prop}
\begin{rem}
 Note that $Z$ has infinite variance because $Y<2$.  The stable distribution was famously used by Mandelbrot\cite{Man63} to model power-like tails and self-similar behaviour in cotton price returns.
\end{rem}
\begin{proof}
Let
\be \label{eq:CGMY_CE}
	\psi(u) =C \Gamma(-Y)((M-iu)^Y+(G+iu)^Y-M^Y-G^Y)+iu{\hat{b}}
\ee
denote the characteristic exponent for the CGMY process.  Then we have
\[
	{\zeta}(u)={\lim_{t \to 0}\exp(t\psi(\frac{u}{t^{1/Y}})) = \exp(-C\Gamma(-Y)|(-i)^Y+i^Y|\,|u|^{Y})}\,,
\]
\nind
where we used that $Y\in(1,2)$. ${\zeta(u)}$ is continuous at zero and we recognize ${\zeta(u)}$ as the characteristic function of a symmetric alpha-stable distribution.  Thus, by L\'{e}vy's convergence theorem (see Theorem 18.1 in Williams\cite{Will91}), the sequence of random variables $(X_t/t^{1/Y})$ converges weakly to $Z$.  The second result follows from the Lemma on page 181, chapter 17 in \cite{Will91}.
\nind \end{proof}

\begin{rem}
Proposition \ref{prop:CGMYConv} is a particular case of a result shown in Rosi\'nski \cite{Ros07} where a more general class of tempered L\'evy measures is considered. Concretely,  \cite{Ros07} considers L\'evy measures of the form
\begin{equation}\label{TSLM}
	\nu(A)=\int_{\mathbb{R}}\int_{0}^{\infty} 1_{A}(u w) u^{-Y-1}e^{-u}du R(dw),
\end{equation}
for a measure $R$ such that $R(\{0\})=0$ and $\int_{\mathbb{R}} (|w|^{2}\wedge |w|^{Y})R(dw)<\infty$. The CGMY model is recovered by taking $R(dw)=C M^{Y}\delta_{\{M^{-1}\}}(dw)+C G^{Y}\delta_{\{-G^{-1}\}}(dw)$. In light of Rosi\'nski's  Theorem 3.1, it follows that Proposition \ref{prop:CGMYConv} also holds true for $Y\in(0,1)$ (finite-variation case) provided that $(X_{t})$ is driftless, i.e. ${\hat{b}}$ in (\ref{CFCGMY}) must be $0$ (otherwise, we have to replace $X_{t}$ by $X_{t}-{\hat{b}}t$).
Note that under $\mathbb{P}^{*}$, $X$ is also driftless {(see Remark \ref{RmrDrftVolEffct})}.
\end{rem}

Another well-know class of L\'evy processes is the Normal Inverse Gaussian (NIG) model, introduced in Barndorff-Nielsen\cite{Bar97}, for which the characteristic function is given by
$$
\mathbb{E}(\exp(iu X_t))=\exp[-t\delta (\sqrt{\alpha^2-(\beta+iu)^2}-\sqrt{\alpha^2-\beta^2})]\,.
$$
The L\'evy density of the NIG model takes the form $\nu(x)=Ce^{Ax}K_{1}(B|x|)/|x|$ where $K_{1}$ is the modified Bessel function of second kind and $A$,$B$, and $C$ are certain positive constants (see \cite{CT04} for their expressions). Hence, one can view the NIG process as an improper tempered stable process in the sense of Rosi\'nski \cite{Ros07}.  It is also easy to see that
$$
\lim_{t \to 0}\mathbb{E}(\exp(iu X_t/t))=\exp[-t\delta |u|]\,.
$$
\nind The right-hand side is the characteristic function of a symmetric alpha-stable random variable $Z$ with $\alpha=1$ and scale parameter $\delta$ i.e. a Cauchy distribution; thus by the same argument we see that $(X_t/t^{1/Y})$ converges weakly to a symmetric Cauchy distribution:
$$
\lim_{t \to 0} \mathbb{P}(X_t/t^{1/Y} > x) =\mathbb{P}(Z > x)\,.
$$

\subsection{At-the-money call option prices for the CGMY model}
\nind Our approach to deal with at-the-money call option prices is based on the following result from Carr\&Madan\cite{CM09}:
\be\label{CMR}
\frac{1}{S_0}\mathbb{E}(S_t-K)_{+}=\mathbb{P}^*(X_t-E> \log \frac{K}{S_0}) \,,\ee

\nind where $E$ is an independent exponential random variable under $\mathbb{P}^*$ with parameter 1.  Now set $K=S_0$.  Consider the CGMY model with $Y\in(1,2)$.  The idea is to use the {small-time, small log-moneyness result in the previous section}.  Indeed, note that
\begin{align}
	t^{-1/Y}\mathbb{P}^{*}(X_{t}\geq{}E)=t^{-1/Y}\int_{0}^{\infty} e^{-x} \mathbb{P}^{*}(X_{t}\geq{}x)dx=\int_{0}^{\infty} e^{-t^{1/Y}u} \mathbb{P}^{*}(X_{t}\geq{}t^{1/Y}u)du.\label{LOP1}
\end{align}	
From our {Proposition \ref{prop:CGMYConv}},
\begin{equation*}
	\mathbb{P}^{*}(X_{t}\geq{}t^{1/Y}u)\to \mathbb{P}^{*}(Z\geq{} u),
\end{equation*}
for any $u>0$, where $Z$ is a symmetric $\alpha$-stable r.v. under $\mathbb{P}^*$.  The previous fact suggests the following result:
\begin{prop}\label{PATM}
Suppose that $X$ is a CGMY {process under $\mathbb{P}$ with $Y\in(1,2)$}. Then, the at-the-money call option price has the following asymptotic behavior:
\begin{align}\label{LOP}
	\lim_{t\to{}0}t^{-1/Y}\mathbb{E}(S_t-S_0)_{+}&=S_{0}\mathbb{E}^{*}(Z_{+}),
\end{align}
where $Z$ is a symmetric {$Y$-stable} r.v. as in Proposition 4.1.
\end{prop}
\begin{proof}
See Appendix \ref{section:Proofs}.
\end{proof}

In order to justify the previous argument, we will need the following estimate:
\begin{lem}\label{KPN}
	Let $X$ {denote a symmetric CGMY process {under $\mathbb{P}$} (hence $G=M$) with $Y\in(1,2)$, {$M>1$, and $C>0$}}. Then, there exists a universal constant $K>0$ such that
	\begin{equation}\label{ESPT}
		\mathbb{P}^{*}\left(X_{t}\geq{}x\right)\leq{} K x^{-Y} t.
	\end{equation}
	for any $t>0$ and $x>0$ satisfying
	\(
		t(b+\int_{|z|\leq{}x/4} z  (e^{z}-1)\nu(dz))< x/4.
	\)
\end{lem}
\begin{proof}
See Appendix \ref{section:Proofs}.
\end{proof}

\begin{rem}\label{GnrClsATM}
	As seen in the proof of Lemma \ref{KPN}, the estimate (\ref{ESPT}) is valid for any pure-jump L\'evy process admitting a symmetric L\'evy density $\nu(x)$ such that
	\[
		\nu(x)\leq C \frac{e^{- M |x|}}{|x|^{1+Y}},
	\]
	for some $Y\in(1,2)$, $C>0$, and $M>1$ . Moreover, as seen in the proofs of Proposition \ref{PATM}\, if we further assume that
	\begin{equation}\label{CTStbl}
		\left(t^{-1/Y} X_{t}\right)_{t} {\;{\stackrel{\frak {D}}{\longrightarrow}}\;} \left(Z_{t}\right)_{t},
	\end{equation}
	as $t\to{}0$ under $\mathbb{P}^{*}$ (for a symmetric $Y$-stable process $(Z_{t})_{t}$), then the asymptotic behavior (\ref{LOP}) will also hold. Condition (\ref{CTStbl}) holds for a wide range of processes (see, for instance, Proposition 1 in \cite{RT11} for relatively mild conditions).
\end{rem}

\subsection{At-the-money implied volatility}

\begin{prop}
For the {CGMY model} with $Y \in (1,2)$ in Proposition \ref{PATM}, we have the following small-time behaviour for the at-the-money implied volatility $\hat{\sigma}_t(0)$
$$\lim_{t \to 0}\hat{\sigma}_t(0)/t^{1/Y-\frac{1}{2}}=\sqrt{2\pi}\,\mathbb{E}^{*}(Z_{+})\,.  $$
\end{prop}
\begin{proof}
We first recall that the dimensionless implied variance $V(t,0)=\hat{\sigma}_t(0)^2 t \to 0$ as $t \to 0$.  Equating prices under the the L\'{e}vy model and the Black-Scholes model, we know that for any $\delta>0$, there exists a $t^*=t^*(\delta)$ such that for all $t<t^*$ we have
$$\mathbb{E}^{*}(Z_{+}) t^{1/Y} (1-\delta) \le \frac{1}{S_0}\mathbb{E}(S_t-S_0)^{+}\le \frac{\sqrt{V(t,0)}}{\sqrt{2 \pi}}(1+\delta)\,.$$
Re-arranging, we see that
$$ \frac{1-\delta}{1+\delta} \,\le \, \frac{\sqrt{V(t,0)}}{\sqrt{2 \pi} \,\mathbb{E}^{*}(Z_{+}) t^{1/Y}}\,.$$
We proceed similarly for the upper bound.
\end{proof}

\section{Robust pricing of variance call options at small maturities}
\label{section:VarianceCallOptions}

Let $(X_t)$ denote the general L\'{e}vy process defined in section \ref{section:SmallTimeLevy}.  The quadratic variation process $[X]_t=\sigma^2 t+\sum_{s \le t} (\Delta X_s)^2$ is a subordinator and has L\'{e}vy density given by
$$q(y)=\frac{\nu(\sqrt{y})}{2\sqrt{y}}+\frac{\nu(-\sqrt{y})}{2\sqrt{y}} \quad \quad \quad  (y>0)$$

\nind (see e.g. \cite{CGMY05}).  The function $f(y)=(y-K)_{+}$ for $K>0$ satisfies the conditions of Theorem 1.1 in Figueroa-L\'{o}pez\cite{FL08}, so we have
\begin{align}
\frac{1}{t}\,\mathbb{E}([X]_t-K)_{+}&= \int_0^{\infty} (y-K)_{+} q(y) dy +O(t)  \quad \quad \quad (t \to 0) \label{VCOAs} \\
&= \int_0^{\infty} (y-K)_{+} \, \big[\frac{\nu(\sqrt{y})}{2\sqrt{y}}+\frac{\nu(-\sqrt{y})}{2\sqrt{y}}\big] dy +O(t) \nn \\
&= \int_{-\infty}^{\infty} (x^2-K)_{+} \nu(x)  \,dx +O(t)\label{VCORO} \\
&=  \frac{1}{t} \,\mathbb{E}(X_t^2-K)_{+} +O(t)\nn \\
&=  \frac{1}{t} \,\mathbb{E}\big[(\ln\frac{S_t}{S_0})^2-K\big]^{+} +O(t) \quad \quad \quad (t \to 0)\,.\nn
\end{align}
 From this we see that an out-of-the-money variance call option of strike $K$ which pays $([X]_t-K)_{+}$ at time $t$ is worth the same as a European-style contract paying $((\ln\frac{S_t}{S_0})^2-K)_{+}$ at time $t$ as $t \to 0$, irrespective of $\nu(\cdot)$.  Note that the diffusion component of $X_t$ does not show up at leading order for small $t$.  We also remark that the higher order terms in (\ref{VCOAs}) and (\ref{VCORO}) can be obtained by using the expansions in Theorem \ref{ThFLH09} and the following identities:
 \begin{align*}
 	\mathbb{E} ([X]_{t}-K)_{+}=\int_{K}^{\infty}\mathbb{P}([X]_{t}\geq{}u) du,
	\quad \mathbb{E} (X_{t}^{2}-K)_{+}=\int_{\sqrt{K}}^{\infty} u \mathbb{P}(X_{t}\geq{}u)du.
 \end{align*}

\section{Numerical examples}\label{Sect:NmrExmpl}

In their seminal work, Carr et al.\cite{CGMY02} {calibrated} the CGMY model and the Variance Gamma (VG) model to option closing prices of several stocks and indices. In this section, we shall use some of their calibrated parameters to illustrate the approximation proposed in this paper. {As in Section \ref{section:SmallTimeLevy}, we are assuming below that the risk-free rate $r$ and the dividend rate $q$ are both set to be zero.}

{Using IBM closing option prices on February 10th, 1999 and maturities of 1 and 2 months, \cite{CGMY02} report the following calibrated parameters for the VG model:
\[
	\sigma=0.4344, \quad \nu=0.1083, \quad \theta=-.3726, \quad \eta=0.0051,
\]
where $\sigma$, $\nu$, and $\theta$ are the three parameters characterizing the VG process (see e.g. \cite{CT04}), and $\eta$ is the volatility of an additional independent Wiener component. In order to assess the accuracy of the call price approximation (\ref{eq:CallOptionAsymptotic}), we have plotted (in Figure \ref{Plot0}) the first and second order approximations of $\mathbb{E}(S_{t}-K)_{+}/t$ as a function of the log moneyness $k=\log K/S_{0}$ for $S_{0}=1$ and time-to-maturities $t=5/252$ and $t=10/252$ (in years). We have also plotted the ``true" option prices obtained via {an inverse Fourier Transform (IFT) method (see Theorem 5.1 in \cite{Lee04} for the case $G=G_1$ corresponding to the call option payoff with $\alpha>0$)}. Table \ref{Table1} also shows the numerical approximations for $1000\times \mathbb{E}(S_{t}-K)_{+}/t$ corresponding to four maturities, together with the numerical values obtained via the IFT. Note that the first order approximation (i.e. $ 1000\times \int_{-\infty}^{\infty} (e^x-e^k)_{+}\nu(x)dx$) is independent of time-to-maturity $t$. {The graphs show that the second order approximation significantly outperforms the first order approximation. The corresponding table shows that the second order approximation is quite good for maturities of 5 to 10 days and logmoneyness values larger than $0.1$.}

The numerical values via the IFT method were implemented in Mathematica, while the coefficient (\ref{ExprCoefd2}) was computed using numerical integration routines of Mathematica.  This computation is typically slow due to the singularity of the L\'evy density $\nu$ and the cumbersome double integrals. A much faster numerical method, valid for bounded variation L\'evy processes, is described in \cite{FL10} (see below for an illustration of this method).

\bs
\begin{figure}[htp]
    {\par \centering
    \includegraphics[width=8.0cm,height=8cm]{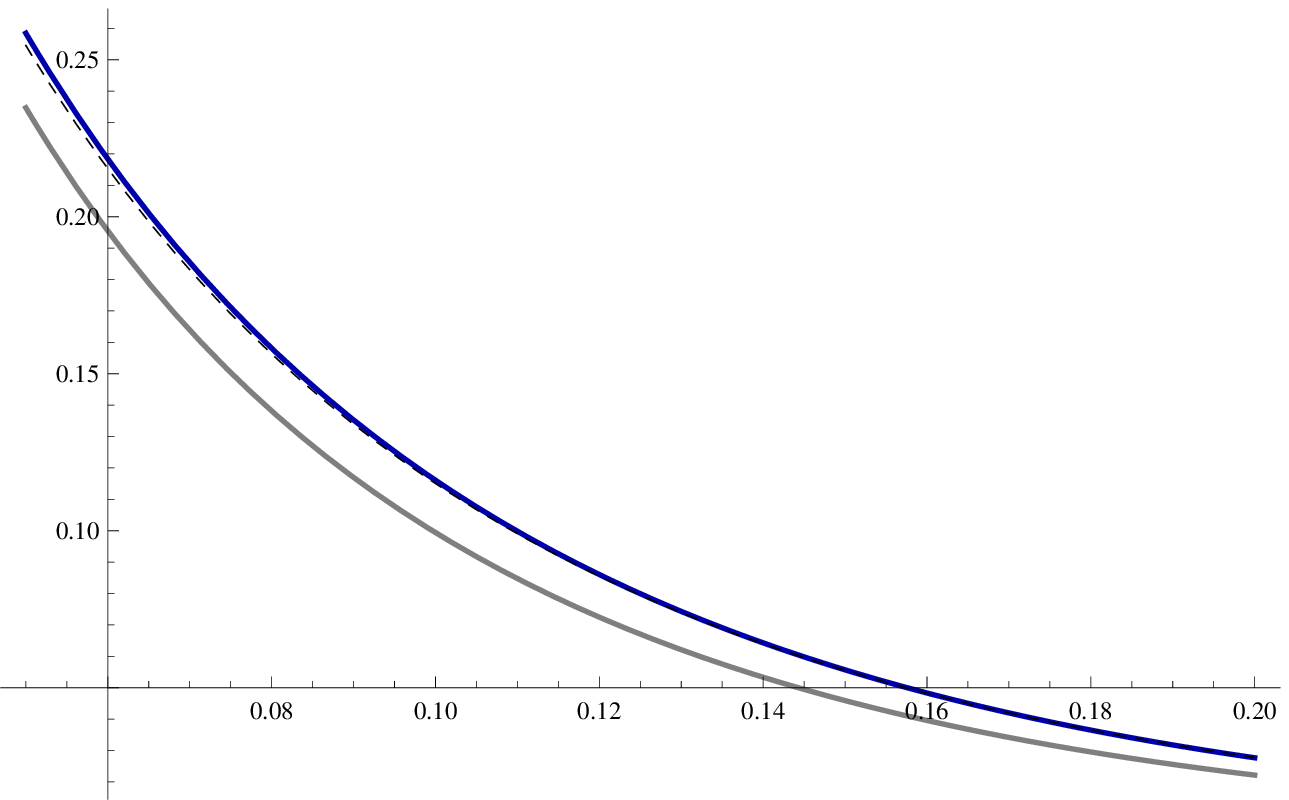}
    \includegraphics[width=8.0cm,height=8cm]{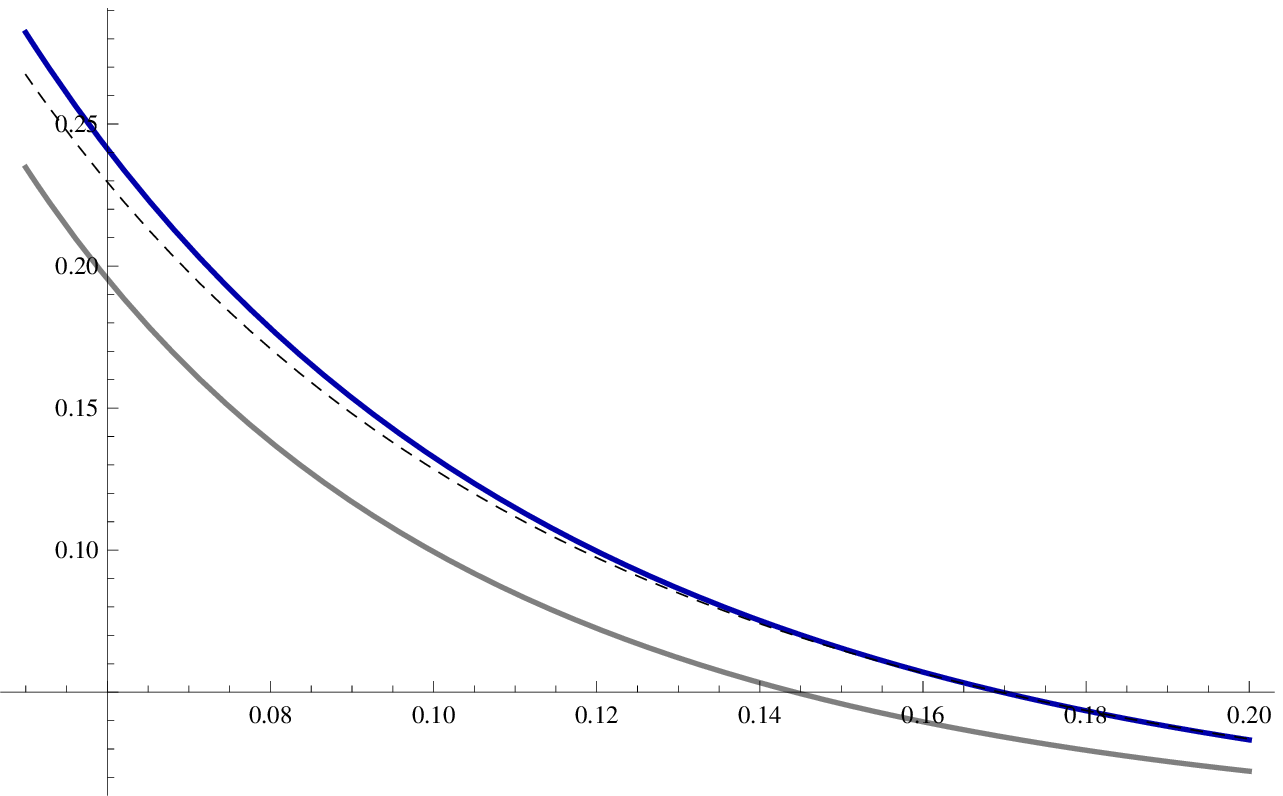}
    \par}
    \caption{Here we have plotted the leading order term (grey line) and the correction term (solid blue line) of the approximation (\ref{eq:CallOptionAsymptotic}) for $\frac{1}{t}\mathbb{E}(S_{t}-K)_{+}$ as a function of the log-moneyness $x=k=\log K/S_0$ for a Variance Gamma model with an independent Brownian component. The parameters of the VG model are $\sigma=0.4344$, $\nu=0.1083$, and $\theta=-.3726$, while the volatility of the independent continuous component is $\eta=0.0051$.  Left  and right panels corresponds to the expiration times $t=5/252$ and $t=10/252$, respectively. The numerical ``true" option prices obtained via the IFT are also shown (dashed grey line).}\label{Plot0}
\end{figure}

\begin{table}
\begin{center}
 \begin{tabular}{|c|c|c|c|c|c|c|c|c|c|}\hline
\multicolumn{2}{|c|}{Time-to-mat. t} &\multicolumn{2}{|c|}{1/252} &\multicolumn{2}{|c|}{5/252} &\multicolumn{2}{|c|}{10/252}&\multicolumn{2}{|c|}{20/252}\\
 \hline
 x & 1st & 2nd & IFT & 2nd & IFT & 2nd & IFT &2nd & IFT\\
\hline
0.05	&	234.6977	&	239.4463	&	239.2843	&	258.4404	&	254.5295	&	282.1831	&	267.3434	&	329.6684	&	277.3445	\\
0.06	&	195.4777	&	200.0560	&	199.9317	&	218.3694	&	215.3264	&	241.2611	&	229.5224	&	287.0445	&	244.4061	\\
0.07	&	163.8997	&	168.2079	&	168.1131	&	185.4408	&	183.0887	&	206.9820	&	197.7644	&	250.0643	&	215.6399	\\
0.08	&	138.1606	&	142.1521	&	142.0805	&	158.1182	&	156.3154	&	178.0757	&	170.8989	&	217.9909	&	190.4486	\\
0.09	&	116.9799	&	120.6392	&	120.5857	&	135.2765	&	133.9099	&	153.5732	&	148.0422	&	190.1665	&	168.3418	\\
0.1	&	99.4165	&	102.7465	&	102.7072	&	116.0661	&	115.0451	&	132.7157	&	128.5074	&	166.0149	&	148.9089	\\
0.11	&	84.7611	&	87.7748	&	87.7466	&	99.8297	&	99.0818	&	114.8984	&	111.7494	&	145.0357	&	131.8027	\\
0.12	&	72.4675	&	75.1840	&	75.1644	&	86.0500	&	85.5170	&	99.6325	&	97.3285	&	126.7974	&	116.7270	\\
0.13	&	62.1087	&	64.5497	&	64.5368	&	74.3137	&	73.9493	&	86.5186	&	84.8855	&	110.9285	&	103.4274	\\
0.14	&	53.3465	&	55.5346	&	55.5269	&	64.2872	&	64.0541	&	75.2279	&	74.1246	&	97.1093	&	91.6844	\\
0.15	&	45.9096	&	47.8674	&	47.8636	&	55.6984	&	55.5669	&	65.4873	&	64.7996	&	85.0649	&	81.3079	\\
0.16	&	39.5787	&	41.3278	&	41.3269	&	48.3238	&	48.2701	&	57.0689	&	56.7045	&	74.5590	&	72.1326	\\
0.17	&	34.1752	&	35.7358	&	35.7372	&	41.9783	&	41.9835	&	49.7815	&	49.6660	&	65.3878	&	64.0145	\\
0.18	&	29.5521	&	30.9433	&	30.9463	&	36.5080	&	36.5571	&	43.4639	&	43.5376	&	57.3758	&	56.8280	\\
0.19	&	25.5884	&	26.8275	&	26.8317	&	31.7841	&	31.8651	&	37.9799	&	38.1947	&	50.3714	&	50.4628	\\
0.2	&	22.1834	&	23.2864	&	23.2913	&	27.6985	&	27.8019	&	33.2136	&	33.5313	&	44.2438	&	44.8227	\\
 \hline \end{tabular}
 \end{center}
 \caption{
Approximations (\ref{eq:CallOptionAsymptotic}) for $1000\times \frac{1}{t}\mathbb{E}(S_{t}-K)_{+}$ as a function of the log-moneyness $x=k=\log K/S_0$ for a Variance Gamma model with an independent Brownian component. The parameters of the VG model are $\sigma=0.4344$, $\nu=0.1083$,  and $\theta=-.3726$, while the volatility of the continuous component is $\eta=0.0051$.  The column ``1st" indicates the first order approximation (which is independent of $t$). The column ``2nd" refers to the second order approximation term.}\label{Table1}
 \end{table}

In order to illustrate the performance of the approximations for larger volatility values, we now consider the parameters:
\[
	\sigma=0.1452,\quad   \theta=-0.1497, \quad  \nu= 0.1536,\quad \eta=0.0869,
\]
which were calibrated to fit INTEL option data as reported in \cite{CGMY02}. The results are shown in Figure \ref{Plot0} for $S_{0}=1$ and time-to-maturities $t=5/252$ and $t=10/252$ (in years). Table \ref{Table2} shows the numerical approximations for $1000\times\mathbb{E}(S_{t}-K)_{+}/t$ corresponding to four maturities. We also show the numerical values obtained via the IFT. {The second order approximation is again quite good for mid-range log-moneyness values and no noticeable difference is observed even though $\eta$ is significantly larger}.

\bs
\begin{figure}[htp]
    {\par \centering
    \includegraphics[width=8.0cm,height=8cm]{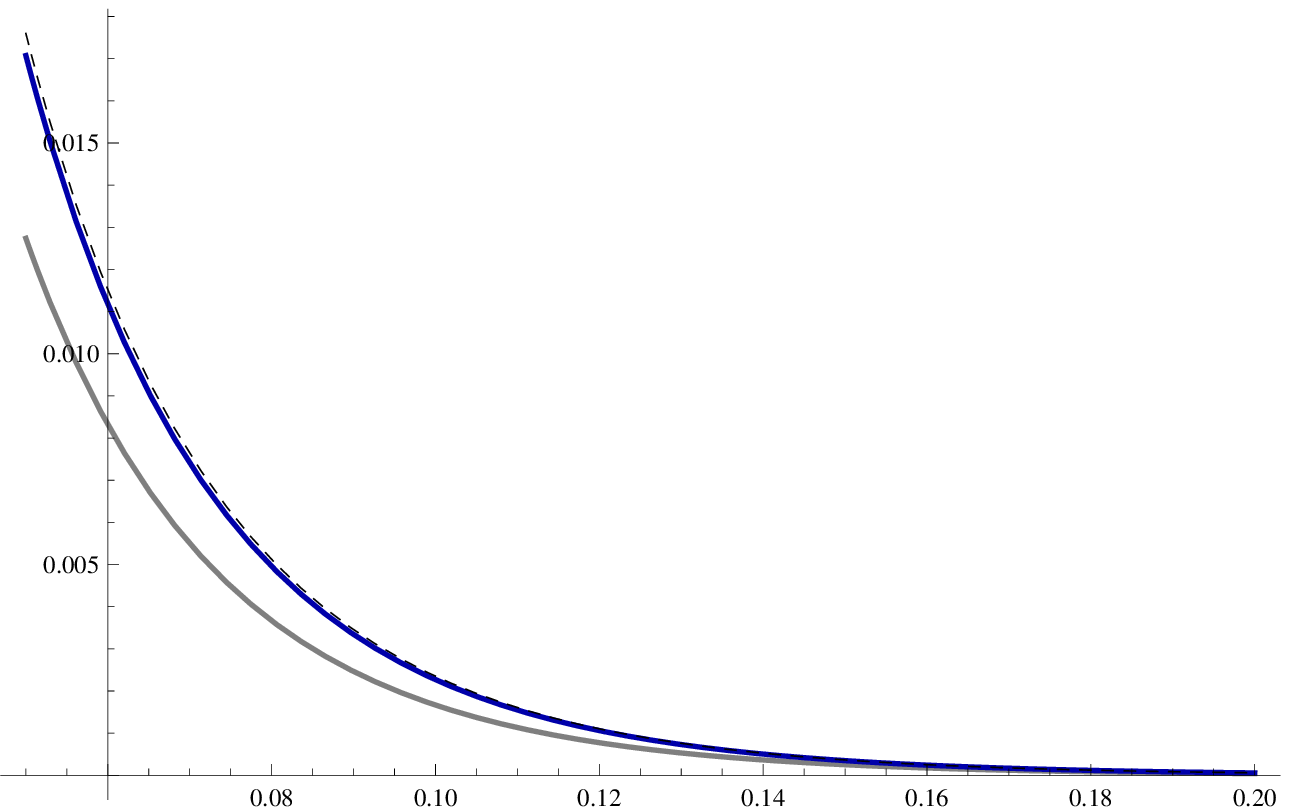}
    \includegraphics[width=8.0cm,height=8cm]{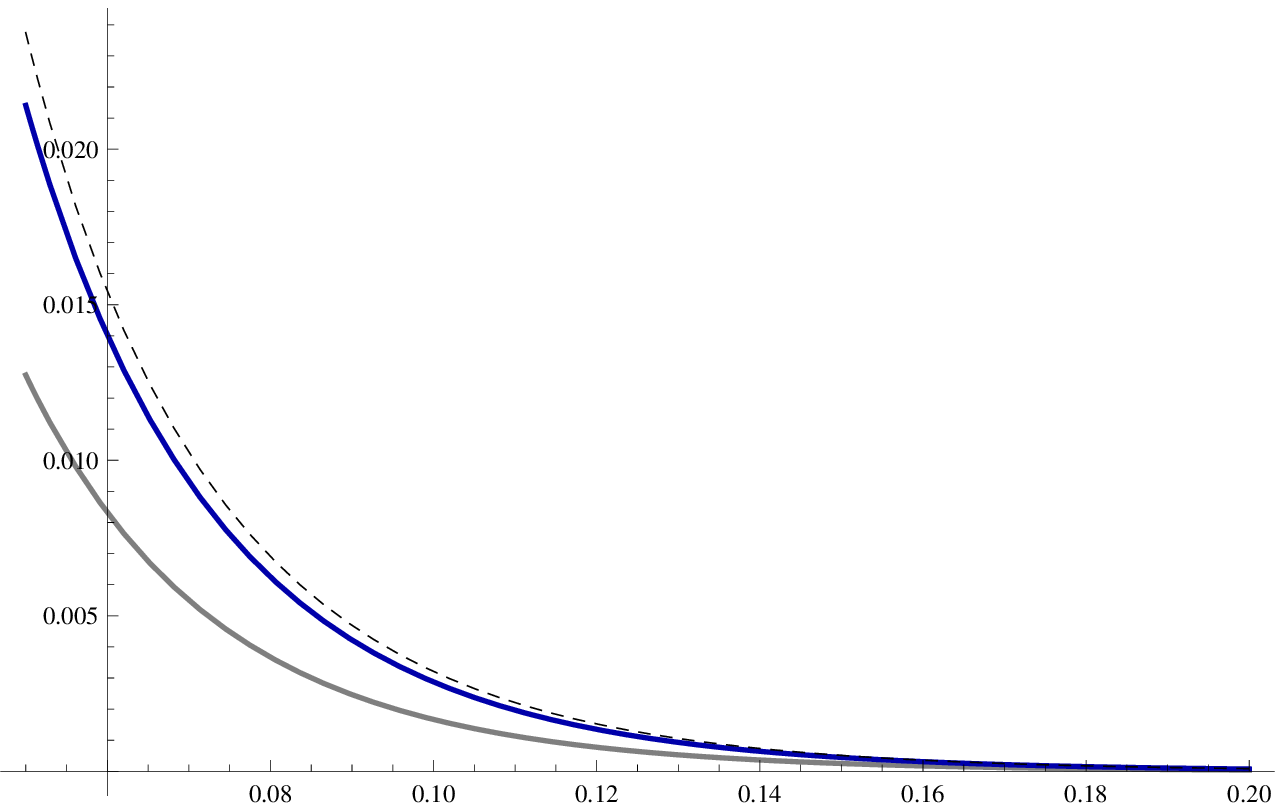}
    \par}
    \caption{Here we have plotted the leading order term (grey line) and the correction term (solid blue line) of the approximation (\ref{eq:CallOptionAsymptotic}) for $\frac{1}{t}\mathbb{E}(S_{t}-K)_{+}$ as a function of the log-moneyness $x=k=\log K/S_0$ for a Variance Gamma model with an independent Brownian component. The parameters of the VG model are $\sigma=0.1452$, $\theta=-0.1497$, $\nu= 0.1536$, while the volatility of the independent continuous component is $\eta=0.0869$.  Left  and right panels corresponds to the expiration times $t=5/252$ and $t=10/252$, respectively. The numerical ``true" option prices obtained via the IFT are also shown (dashed grey line).
    }\label{Plot1}
\end{figure}

\begin{table}
\begin{center}
 \begin{tabular}{|c|c|c|c|c|c|c|c|c|c|}\hline
\multicolumn{2}{|c|}{Time-to-mat. t} &\multicolumn{2}{|c|}{1/252} &\multicolumn{2}{|c|}{5/252} &\multicolumn{2}{|c|}{10/252}&\multicolumn{2}{|c|}{20/252}\\
 \hline
 x & 1st & 2nd & IFT & 2nd & IFT & 2nd & IFT &2nd & IFT\\
\hline
0.05	&	12.7382	&	13.6052	&	13.6253	&	17.0732	&	17.5978	&	21.4081	&	23.7455	&	30.0780	&	36.6508	\\
0.06	&	8.3203	&	8.8906	&	8.9038	&	11.1717	&	11.5085	&	14.0232	&	15.4255	&	19.7261	&	24.6815	\\
0.07	&	5.4984	&	5.8797	&	5.8887	&	7.4046	&	7.6352	&	9.3108	&	10.2499	&	13.1232	&	16.7357	\\
0.08	&	3.6672	&	3.9249	&	3.9312	&	4.9559	&	5.1175	&	6.2446	&	6.9034	&	8.8221	&	11.4468	\\
0.09	&	2.4641	&	2.6398	&	2.6443	&	3.3426	&	3.4572	&	4.2212	&	4.6912	&	5.9782	&	7.8929	\\
0.1	&	1.6660	&	1.7865	&	1.7897	&	2.2687	&	2.3504	&	2.8714	&	3.2090	&	4.0769	&	5.4783	\\
0.11	&	1.1323	&	1.2154	&	1.2177	&	1.5479	&	1.6063	&	1.9635	&	2.2067	&	2.7947	&	3.8216	\\
0.12	&	0.7730	&	0.8306	&	0.8322	&	1.0608	&	1.1027	&	1.3485	&	1.5239	&	1.9241	&	2.6763	\\
0.13	&	0.5298	&	0.5698	&	0.5709	&	0.7297	&	0.7598	&	0.9297	&	1.0562	&	1.3295	&	1.8801	\\
0.14	&	0.3643	&	0.3922	&	0.3930	&	0.5037	&	0.5252	&	0.6430	&	0.7343	&	0.9216	&	1.3240	\\
0.15	&	0.2513	&	0.2708	&	0.2714	&	0.3486	&	0.3641	&	0.4460	&	0.5119	&	0.6406	&	0.9344	\\
0.16	&	0.1738	&	0.1875	&	0.1879	&	0.2420	&	0.2531	&	0.3101	&	0.3577	&	0.4464	&	0.6607	\\
0.17	&	0.1205	&	0.1301	&	0.1304	&	0.1683	&	0.1763	&	0.2161	&	0.2504	&	0.3117	&	0.4679	\\
0.18	&	0.0837	&	0.0905	&	0.0907	&	0.1173	&	0.1231	&	0.1509	&	0.1757	&	0.2181	&	0.3318	\\
0.19	&	0.0583	&	0.0630	&	0.0632	&	0.0819	&	0.0861	&	0.1056	&	0.1234	&	0.1528	&	0.2356	\\
0.2	&	0.0407	&	0.0440	&	0.0441	&	0.0573	&	0.0603	&	0.0740	&	0.0869	&	0.1073	&	0.1675	\\
 \hline \end{tabular}
 \end{center}
 \caption{
Approximations (\ref{eq:CallOptionAsymptotic}) for $1000\times \frac{1}{t}\mathbb{E}(S_{t}-K)_{+}$ as a function of the log-moneyness $x=k=\log K/S_0$ for a Variance Gamma model with an independent Brownian component. The parameters of the VG model are $\sigma=0.1452$, $\theta=-0.1497$, $\nu= 0.1536$, while the volatility of the continuous component is $\eta=0.0869$.  The column ``1st" indicates the first order approximation (which is independent of $t$). The column ``2nd" refers to the second order approximation term.     }\label{Table2}
 \end{table}

For the case of Microsoft option prices on December 9th, 1999 and maturities of 1 and 2 months, \cite{CGMY02} report the following parameters for a CGMY model:
\[
	C=1.1, \quad G=5.09, \quad M=8.6, \quad Y=0.4456.
\]
Table \ref{Table3} shows the numerical approximations for $1000\times \mathbb{E}(S_{t}-K)_{+}/t$ corresponding to four maturities, together with the numerical values obtained via the IFT (computed using Mathematica).  {As before, the approximations perform quite well and we are able to attain a decent approximation even for a maturity of 20 days}. To compute the second order approximations (or more specifically, to compute the coefficient (\ref{ExprCoefd2})), we have employed the method in \cite{FL10}.

\bs
{\small
\begin{table}
\begin{center}
 \begin{tabular}{|c|c|c|c|c|c|c|c|c|c|}\hline
\multicolumn{2}{|c|}{Time-to-mat. t} &\multicolumn{2}{|c|}{1/252} &\multicolumn{2}{|c|}{5/252} &\multicolumn{2}{|c|}{10/252}&\multicolumn{2}{|c|}{20/252}\\
 \hline
 x & 1st & 2nd & IFT & 2nd & IFT & 2nd & IFT &2nd & IFT\\
\hline
0.05	&	118.8662	&	120.2883	&	120.5386	&	125.9768	&	125.9179	&	133.0875	&	131.5844	&	147.3088	&	139.5891	\\
0.06	&	99.6004	&	100.8808	&	101.1351	&	106.0023	&	106.0868	&	112.4042	&	111.5177	&	125.2081	&	119.9024	\\
0.07	&	84.3149	&	85.4610	&	85.7023	&	90.0455	&	90.1924	&	95.7760	&	95.2726	&	107.2372	&	103.5827	\\
0.08	&	71.9095	&	72.9321	&	73.1727	&	77.0226	&	77.2339	&	82.1358	&	81.9201	&	92.3620	&	89.9114	\\
0.09	&	61.7191	&	62.6303	&	62.8747	&	66.2750	&	66.5275	&	70.8309	&	70.8150	&	79.9426	&	78.3608	\\
0.1	&	53.2682	&	54.0799	&	54.3141	&	57.3264	&	57.5892	&	61.3846	&	61.4910	&	69.5011	&	68.5328	\\
0.11	&	46.1664	&	46.8892	&	47.1192	&	49.7805	&	50.0626	&	53.3947	&	53.6011	&	60.6229	&	60.1205	\\
0.12	&	40.1763	&	40.8204	&	41.0433	&	43.3967	&	43.6782	&	46.6171	&	46.8806	&	53.0579	&	52.8833	\\
0.13	&	35.0705	&	35.6445	&	35.8690	&	37.9408	&	38.2302	&	40.8111	&	41.1241	&	46.5517	&	46.6292	\\
0.14	&	30.7034	&	31.2154	&	31.4361	&	33.2632	&	33.5566	&	35.8230	&	36.1693	&	40.9425	&	41.2037	\\
0.15	&	26.9570	&	27.4140	&	27.6311	&	29.2418	&	29.5285	&	31.5266	&	31.8864	&	36.0962	&	36.4806	\\
0.16	&	23.7163	&	24.1244	&	24.3391	&	25.7565	&	26.0433	&	27.7968	&	28.1703	&	31.8772	&	32.3565	\\
0.17	&	20.9085	&	21.2731	&	21.4858	&	22.7315	&	23.0167	&	24.5545	&	24.9355	&	28.2005	&	28.7454	\\
0.18	&	18.4722	&	18.7982	&	19.0082	&	20.1025	&	20.3798	&	21.7327	&	22.1107	&	24.9933	&	25.5756	\\
0.19	&	16.3432	&	16.6349	&	16.8407	&	17.8017	&	18.0761	&	19.2602	&	19.6377	&	22.1771	&	22.7868	\\
0.2	&	14.4852	&	14.7463	&	14.9482	&	15.7910	&	16.0580	&	17.0968	&	17.4672	&	19.7084	&	20.3280	\\
0.21	&	12.8531	&	13.0870	&	13.2891	&	14.0226	&	14.2859	&	15.1920	&	15.5580	&	17.5310	&	18.1563	\\
0.22	&	11.4193	&	11.6289	&	11.8268	&	12.4672	&	12.7267	&	13.5150	&	13.8752	&	15.6108	&	16.2344	\\
0.23	&	10.1595	&	10.3474	&	10.5434	&	11.0990	&	11.3517	&	12.0385	&	12.3891	&	13.9176	&	14.5312	\\
0.24	&	9.0459	&	9.2145	&	9.4085	&	9.8885	&	10.1371	&	10.7310	&	11.0744	&	12.4161	&	13.0193	\\
0.25	&	8.0621	&	8.2133	&	8.4040	&	8.8179	&	9.0625	&	9.5737	&	9.9096	&	11.0853	&	11.6753	\\
0.26	&	7.1931	&	7.3287	&	7.4365	&	7.8714	&	8.1099	&	8.5498	&	8.8759	&	9.9065	&	10.4792	\\
0.27	&	6.4212	&	6.5430	&	6.7291	&	7.0301	&	7.2645	&	7.6389	&	7.9573	&	8.8567	&	9.4132	\\
0.28	&	5.7374	&	5.8468	&	5.8054	&	6.2842	&	6.5132	&	6.8309	&	7.1400	&	7.9243	&	8.4622	\\
0.29	&	5.1285	&	5.2267	&	5.4878	&	5.6194	&	5.8445	&	6.1103	&	6.4118	&	7.0920	&	7.6128	\\
0.3	&	4.5867	&	4.6749	&	4.8038	&	5.0275	&	5.2487	&	5.4683	&	5.7624	&	6.3499	&	6.8534	\\
0.31	&	4.1050	&	4.1842	&	3.4559	&	4.5009	&	4.7173	&	4.8968	&	5.1826	&	5.6886	&	6.1739	\\
0.32	&	3.6746	&	3.7457	&	3.7292	&	4.0301	&	4.2427	&	4.3856	&	4.6643	&	5.0966	&	5.5652	\\
0.33	&	3.2905	&	3.3543	&	3.6098	&	3.6097	&	3.8185	&	3.9289	&	4.2006	&	4.5673	&	5.0195	\\
0.34	&	2.9479	&	3.0053	&	3.2470	&	3.2346	&	3.4391	&	3.5212	&	3.7855	&	4.0944	&	4.5299	\\
0.35	&	2.6410	&	2.6925	&	2.8716	&	2.8983	&	3.0991	&	3.1555	&	3.4134	&	3.6701	&	4.0903	\\
 \hline \end{tabular}
 \end{center}
 \caption{
Approximations (\ref{eq:CallOptionAsymptotic}) for $1000\times \frac{1}{t}\mathbb{E}(S_{t}-K)_{+}$ as a function of the log-moneyness $x=k=\log K/S_0$ for the CGMY model with parameter values $C=1.1$, $G=5.09$, $M=8.6$, and $Y=0.4456$.  The column ``1st" indicates the first order approximation (which is independent of $t$). The column ``2nd" refers to the second order approximation term.    }\label{Table3}
 \end{table}
}}

We now proceed to illustrate the performance of the implied volatility approximations described in Section \ref{SSec:ImplVol}. Concretely, we analyze the relative error of the approximations
\be
\label{eq:ImpliedVolatilityLevyv2}
\tilde{\sigma}_{t,1}(k)= \sqrt{\frac{V_{0}(t,k)}{t}},\quad
\tilde{\sigma}_{t,2}(k)= \sqrt{\frac{V_{0}(t,k)(1+V_{1}(t,k))}{t}}.
\ee
Let us first analyze the Variance Gamma model with parameter values as above. The left panel of Figure \ref{Plot3} shows the relative errors $(\tilde{\sigma}_{t,1}-\hat\sigma_{t})/\hat\sigma_{t}$ and $(\tilde{\sigma}_{t,2}-\hat\sigma_{t})/\hat\sigma_{t}$ as a function of time-to-maturity $t$ for values of $k$ ranging from $0.1$ to $0.3$. Note that both $\tilde\sigma_{t,1}$ and $\tilde\sigma_{t,2}$ consistently underestimate the true implied volatility. For $k=0.3$, the first order approximation is actually quite good with a relative error of about $-5\%$ uniformly in $t$ and it is only for very small values $t$ (less than $3$ days) when $\tilde{\sigma}_{t,2}$  is better than $\tilde{\sigma}_{t,1}$. However, for the other values of $k$, $\tilde{\sigma}_{t,2}$  significantly outperforms $\tilde{\sigma}_{t,1}$. For instance, for $k=0.2$, the relative error of  $\tilde\sigma_{t,1}$ ranges from $-19\%$ to $-34\%$ with a mean absolute error of $27.0\%$, while the relative error of  $\tilde\sigma_{t,2}$ rages from $-4.4\%$ to $-23\%$ with a mean absolute error of $14.2\%$. The left panel of Figure \ref{Plot5} compares the term structure of the approximated implied volatilities to the ``true" implied volatility\footnote{The ``true" implied volatility is actually an approximation as we apply numerical integration to compute $\mathbb{E}(S_{T}-K)_{+}$ using the closed-form density of the VG process. This approximation seems not to be very accurate for $t$ smaller than $3$ days.}.  The right panel of Figure \ref{Plot3} shows the analog results for the CGMY with parameter values as above. The results are qualitatively similar to those of the Variance Gamma model. However, all the approximations seem to perform better in terms of error stability in time and accuracy. For $k=0.2$, the relative error of $\tilde\sigma_{t,1}$ ranges from $-12\%$ to $-20\%$  with a mean absolute error of $18.6\%$, while the relative error of $\tilde\sigma_{t,2}$ ranges from $0.83\%$ to $-13\%$  with a mean absolute error of $9.25\%$. The right panel of Figure \ref{Plot5} compares the term structure of the approximated implied volatilities to the ``true" implied volatility\footnote{The true implied volatility is computed by integrating numerically the the density of the CGMY model, which itself is obtained by Fast Fourier methods.}.

 \begin{figure}[htp]
    {\par \centering
    \includegraphics[width=8.0cm,height=8cm]{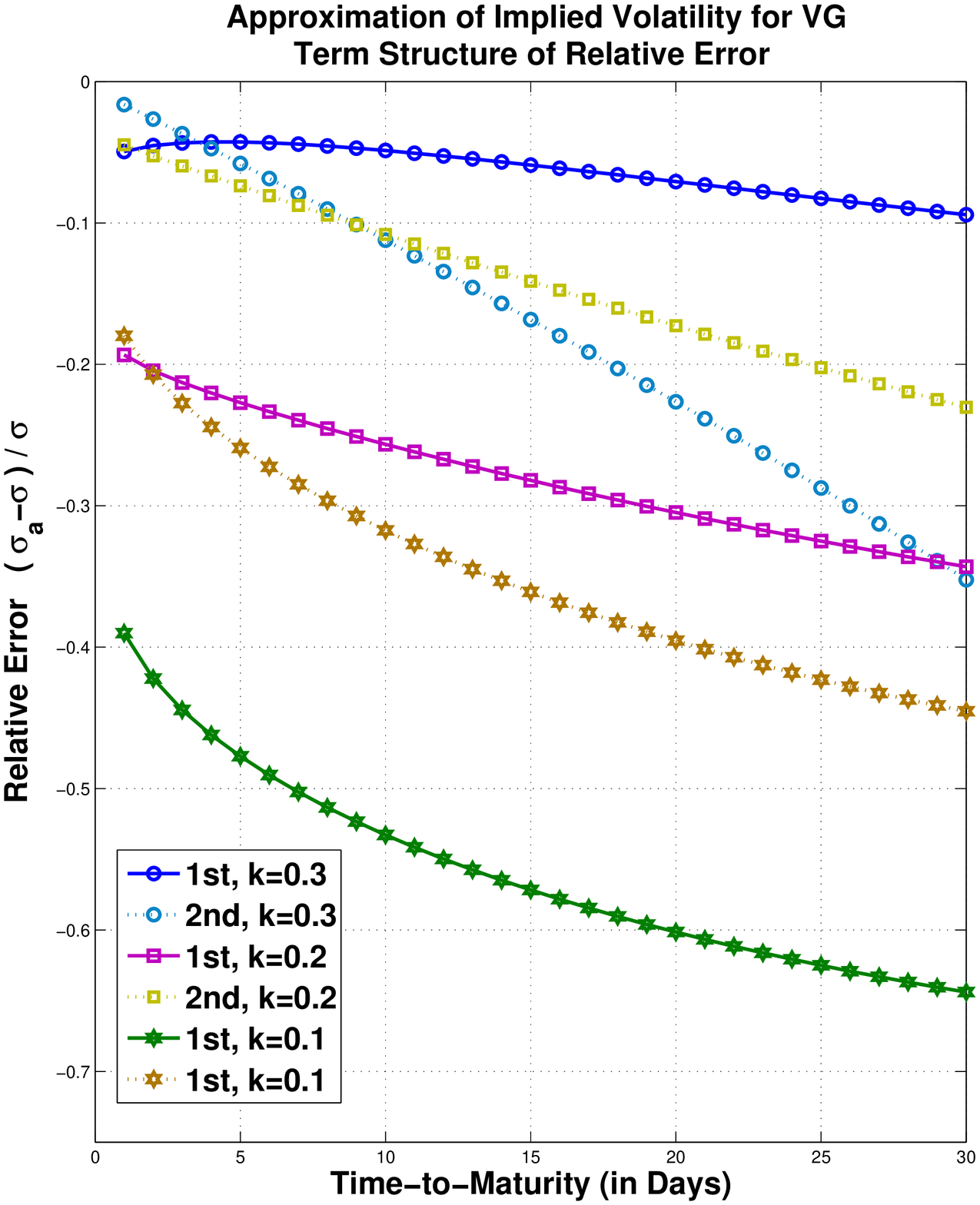}
    \includegraphics[width=8.0cm,height=8cm]{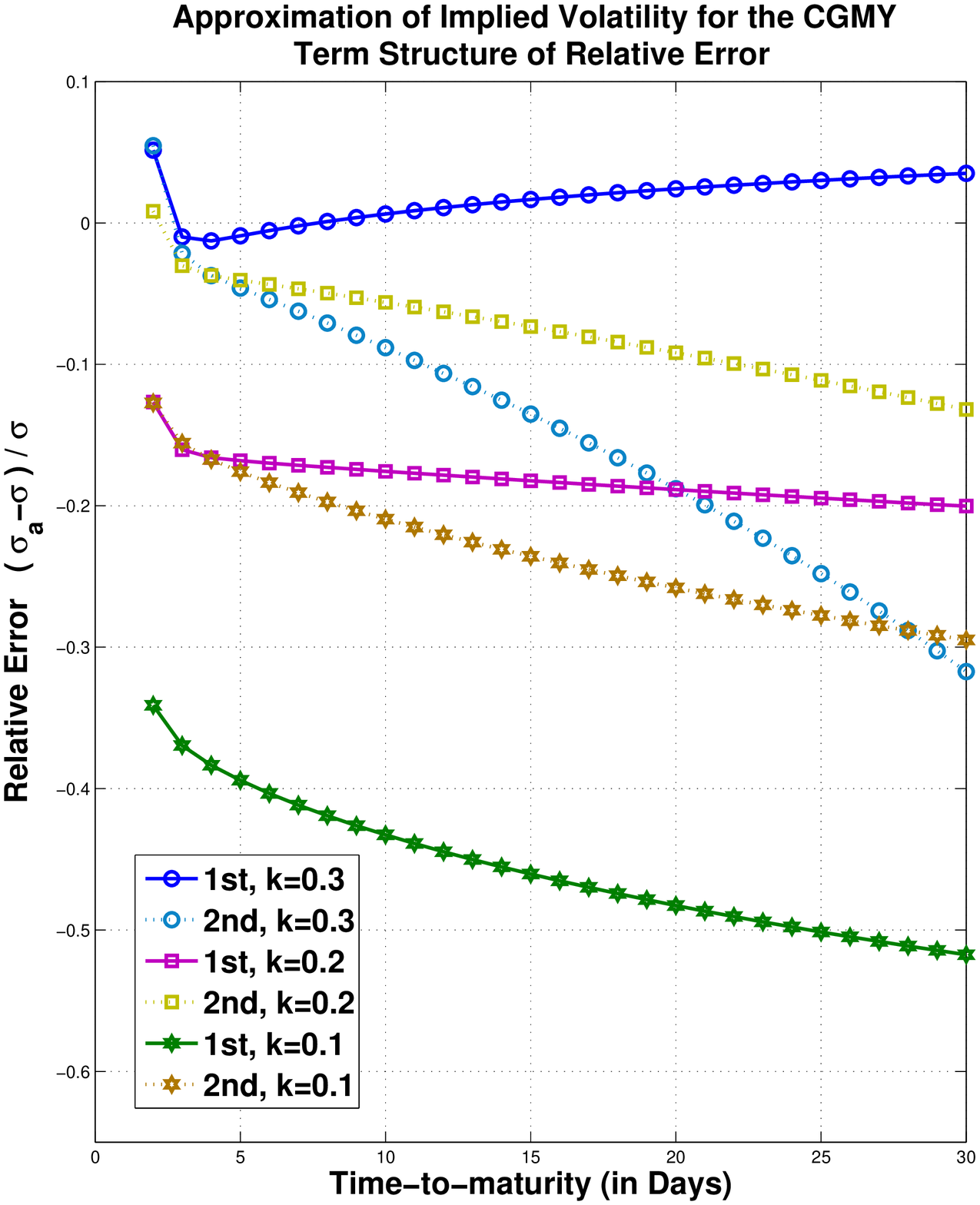}
    \par}\vspace{-.8 cm}
    \caption{Relative errors of the implied volatility approximations for the VG and CGMY models as function of time to maturity {using} the two estimators $\tilde{\sigma}_{t,1}$ and $\tilde{\sigma}_{t,2}$ in (\ref{eq:ImpliedVolatilityLevyv2}).
    }\label{Plot3}
\end{figure}
 \begin{figure}[htp]
    {\par \centering
  \includegraphics[width=8.0cm,height=8cm]{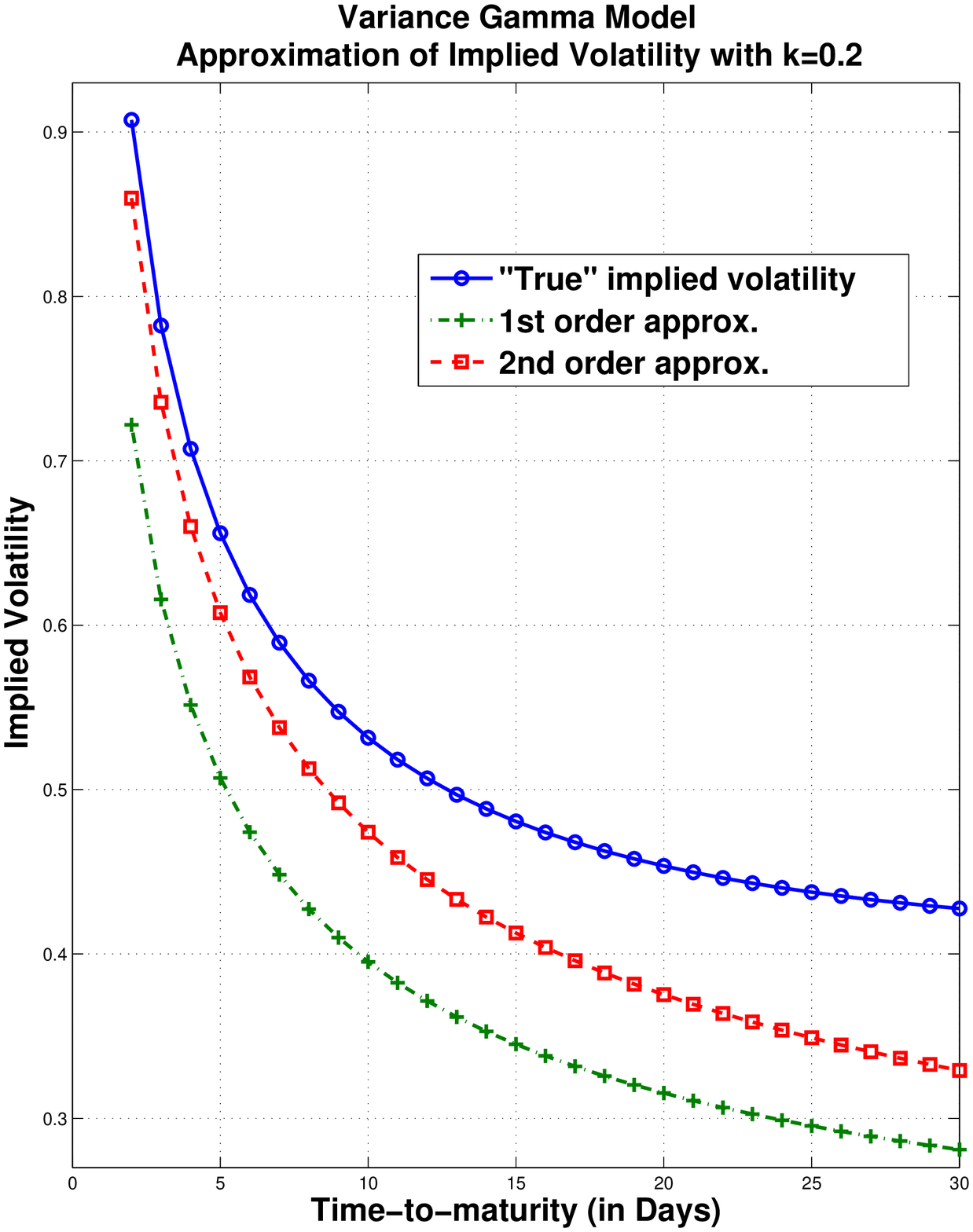}
    \includegraphics[width=8.0cm,height=8cm]{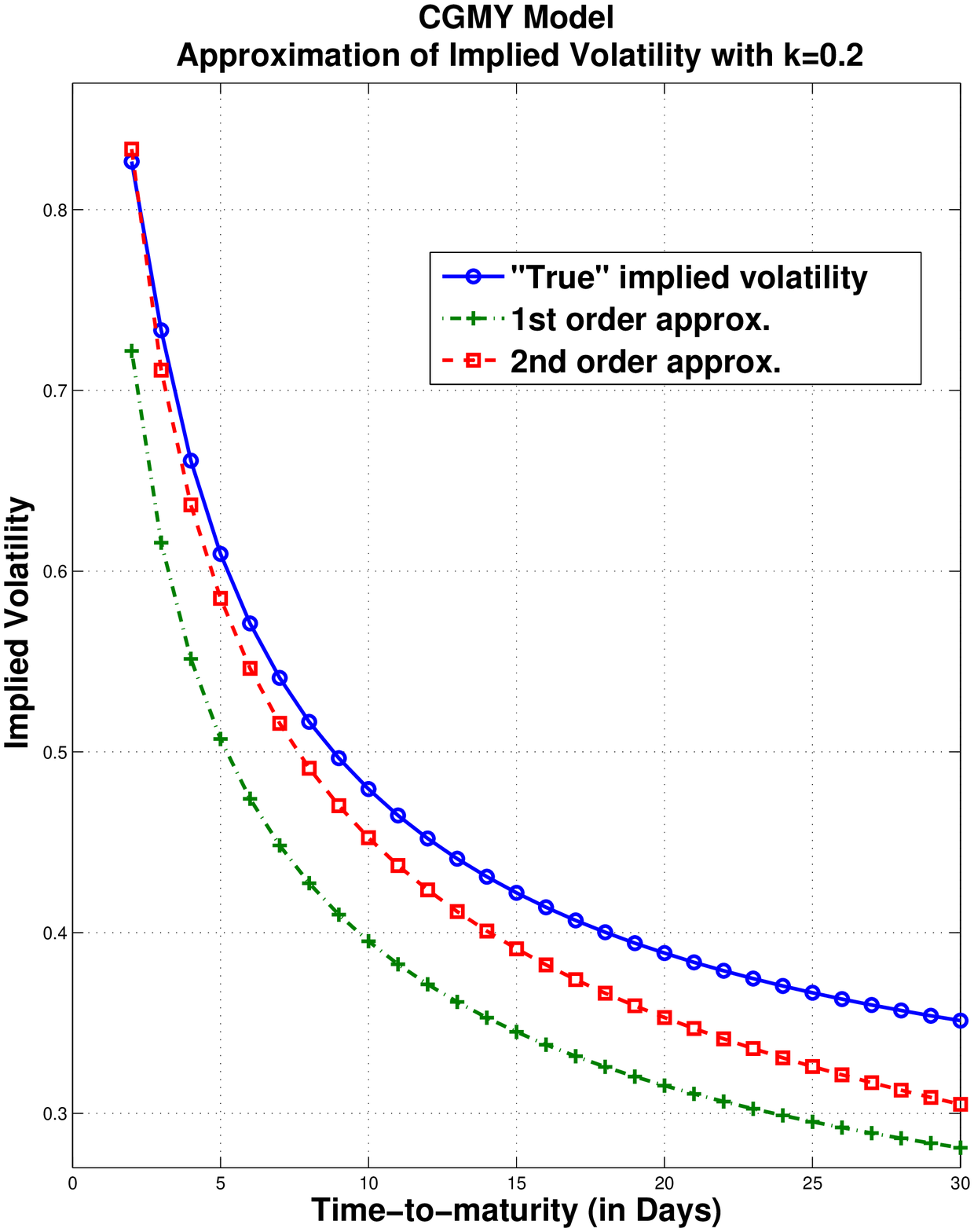}
    \par}\vspace{-.8 cm}
    \caption{Term structure of implied volatility approximations for the Variance Gamma model (left panel) and the CGMY model (right panel) using the estimators $\tilde{\sigma}_{t,1}$ and $\tilde{\sigma}_{t,2}$ in (\ref{eq:ImpliedVolatilityLevyv2}).
    }\label{Plot5}
\end{figure}


\vspace{.5 cm}
\noindent
{\bf Acknowledgments:}
It is a pleasure to thank Michael Roper for pointing up several notational mistakes and other helpful comments. We would also like to thank Peter Tankov for providing us with his manuscript.

\appendix
\renewcommand{\theequation}{A-\arabic{equation}}
\setcounter{equation}{0}  
\section{Proofs}  
\label{section:Proofs} 

\begin{proof}[Proof of Theorem \ref{thm:LevyImpliedVariance}]

\nind We know that $V(t,k) \to 0$.  Equating call prices in the small-time limit under the exponential L\'{e}vy model (using Proposition \ref{prop:SmallTimeCallsLevy}), and the Black-Scholes model with zero interest rates and implied variance $V=V(t,k)$ (using e.g. {Proposition 3.4} in \cite{FJL10} or Lemma 2.5 in \cite{GHLOW09}) we know that for any $\delta>0$, there exists a $t^*=t^*(\delta)$ such that for all $t<t^*$
\bq t a_0(k)(1-\delta)  \, \le \, \frac{1}{S_0}\mathbb{E}(S_t-K)^{+} \,\le \,
e^{-\frac{1}{2}k^2(1-\delta)/V(t,k)} \,.\eq

Re-arranging, we see that
$$ -V(t,k) \log[t a_0(k)(1-\delta)]  \, \ge \,
\frac{1}{2}k^2(1-\delta) \,,$$
\nind or
$$ V(t,k) \cdot \log(\frac{1}{t})  \, \ge \,
\frac{1}{2}k^2(1-\delta) + V(t,k)\log(1-\delta) +V(t,k)\log a_0(k) \,.$$

$V(t,k) \to 0$, so this yields a lower bound for $V(t,k)$.  Using a similar argument for the corresponding upper bound, we establish the leading order asymptotic behaviour for the implied variance as
\be
V(t,k) \,\sim V_0(t,k):= \frac{\frac{1}{2}k^2}{\log (\frac{1}{t})} \quad \quad \quad \quad \quad (t \to 0)\,. \ee
\nind Now let
$
V(t,k)=V_0(t,k)\big[1+\tilde{V}_1(t,k)\big]\,
$
\nind
and note that $\tilde{V}_1(t,k)=o(1)$ as $t\to{}0$.  Then for any $\delta>0$, there exists a $t^{**}=t^{**}(\delta)$ such that for $t<t^{**}$ we have
\be
\frac{1}{t}[\frac{1}{t}\,\frac{1}{S_0}\mathbb{E}(S_t-K)^{+}-a_0(k)]-a_1(k) \ge -\delta\,.
\ee
Re-arranging, we have
\be
t a_0(k) + (a_1(k)-\delta)t^2 \,\le\, \frac{1}{S_0}\mathbb{E}(S_t-K)^{+} \,.
\ee
Using this bound and again equating small-time call prices under the L\'{e}vy model and the Black-Scholes model, we have that there exists a positive constant $c$ such that for $t$ small enough
\bq
  t a_0(k) + (a_1(k)-\delta)t^2 \, \le \, \frac{1}{S_0}\mathbb{E}(S_t-K)^{+} &\le & \,  \frac{e^{\frac{1}{2}k}V(t,k)^{\frac{3}{2}}}{\sqrt{2\pi} \,k^2}e^{-\frac{1}{2}k^2/V(t,k)}(1+ c V(t,k))\nn \\
 &=& \frac{e^{\frac{1}{2}k}V_0(t,k)^{\frac{3}{2}}(1+\tilde{V}_1(t,k))^{\frac{3}{2}}}{\sqrt{2\pi} \,k^2}e^{-\frac{1}{2}k^2/\{V_0(t,k)(1+\tilde{V}_1(t,k))\}}(1+c V(t,k)) \quad \quad \quad  \nn \\
 &\le& \frac{e^{\frac{1}{2}k}V_0(t,k)^{\frac{3}{2}}}{\sqrt{2\pi}\,k^2}e^{-\frac{1}{2}k^2/\{V_0(t,k)(1+\tilde{V}_1(t,k))\}}(1+\mathcal{E}(t,k)), \quad \quad \quad \quad  \quad (t \to 0)\,.\nn
\eq
where $\mathcal{E}(t,k):=(1+\tilde{V}_{1}(t,k))^{3/2}(1+c V(t,k))-1$, which converges to $0$ as $t\to{}0$.
Dividing both sides by $t=e^{-\frac{1}{2}k^2/V_0(t,k)}$ we have
$$
a_{0}(k) +(a_1(k)-\delta)t \,\le \, \frac{e^{\frac{1}{2}k}V_0(t,k)^{\frac{3}{2}}}{\sqrt{2\pi}\, k^2}
\,e^{\frac{1}{2}k^2 \tilde{V}_1(t,k)/\{V_0(t,k)(1+\tilde{V}_1(t,k))\}}(1+\mathcal{E}(t,k)) \,,$$
\nind and re-arranging we obtain
\bq
\frac{\tilde{V}_1(t,k)}{1+\tilde{V}_1(t,k)} \,&\ge & \, \frac{2}{k^2}V_0(t,k) \log\big[ (a_0(k) + (a_1(k)-\delta)t)\sqrt{2\pi}\,k^2 e^{-\frac{1}{2}k}V_0(t,k)^{-\frac{3}{2}}/(1+\mathcal{E}(t,k))\big]\,.\nn \\
 \,&= & \, \frac{2}{k^2}V_0(t,k) \log\big[ (a_0(k)+t\,a_1(k))\sqrt{2\pi}\,k^2 e^{-\frac{1}{2}k}V_0(t,k)^{-\frac{3}{2}}\big] + \frac{2}{k^2}V_0(t,k) \log\big[\frac{1-\frac{\delta t}{a_0(k)+a_1(k) t}}{1+\mathcal{E}(t,k)}\big]\nn \\
 \,&= & \underbrace{\frac{1}{\log(\frac{1}{t})}\log\left[\frac{4 \sqrt{\pi}a_{0}(k)e^{-k/2}}{k}[\log(\frac{1}{t})]^{\frac{3}{2}}\right]}_{V_{1}(t,k)}
 + \underbrace{\frac{1}{\log(\frac{1}{t})}\log\left\{\left[ 1+t\,\frac{a_1(k)}{a_0(k)}\right]
 \left[\frac{1-\frac{\delta t}{a_0(k)+a_1(k) t}}{1+\mathcal{E}(t,k)}\right]\right\}}_{\mathcal{E}'(t,k)}\,.\nn
\eq
Note that $V_1=V_1(t,k)=O\left(\frac{\log\log \frac{1}{t}}{\log \frac{1}{t}}\right)$ and $\mathcal{E}'=\mathcal{E}'(t,k)=o(\frac{1}{\log\frac{1}{t}})$ since $\mathcal{E}(t,k)\to{}0$ as $t\to{}0$.
\nind Solving the inequality $\frac{\tilde{V}_1}{1+\tilde{V}_1} \ge V_1 +\mathcal{E}' $, we find that
$$
\tilde{V_1} \,\ge \, \frac{V_1+\mathcal{E}'}{1-(V_1+\mathcal{E}')}
=V_1 + \mathcal{E}'  +\frac{V_1^2 +2 V_1  \mathcal{E}' + \mathcal{E}'{}^{2}}{1-(V_1+\mathcal{E}')}.
$$
Since $\mathcal{E}'(t,k)>V^2_1(t,k)$ for $t$ sufficiently small, we conclude that
$\tilde{V_1} \,\ge V_{1}+o(\frac{1}{\log\frac{1}{t}})$. Proceeding similarly for the upper bound, we conclude that
\[
	\tilde{V}_1(t,k) = V_1(t,k) + o(\frac{1}{\log\frac{1}{t}})\,.
\]
as $t\to0$.
\end{proof}

\medskip
\noindent\begin{proof}[Proof of Theorem \ref{TAB2}]
	Let $\bar{F}(t):=\mathbb{P}(X_{t}\geq{}x)$ and $B:=\nu[x,\infty)+\sup_{t>0} \mathbb{P}(X_{t}\geq{}x)/t$.
	In the light of Theorem \ref{ThFLH09}, there exist constants $t_{0}>0$ and $K<\infty$ such that
	 \[
	 	\left|\frac{1}{t} \mathbb{P}(X_{t}\geq{}x)-\nu[x,\infty)\right|\leq K t,
	\]
	for any $0<t<t_{0}$.  Next, conditioning on $T_{t}$,
	\begin{align*}
		\frac{1}{t}\,\mathbb{P}(X_{T_{t}}\geq{}x)=
		\frac{1}{t}\,\mathbb{E}\, \bar{F}(T_{t})=\frac{\nu[x,\infty)}{t}\,\mathbb{E}\,T_{t}+\frac{1}{t}\,\mathbb{E}\left(\left\{\frac{1}{T_{t}}\bar{F}(T_{t})-\nu[x,\infty)\right\}T_{t}\right).
	\end{align*}
	Let $R_2(t)$ denote the second term on the right-hand side, which we can bound as follows:
	\begin{align*}
		|R_{2}|&\leq \frac{1}{t}\,\mathbb{E}\left(1_{\{T_{t}<t_{0}\}}\left|\frac{1}{T_{t}}\bar{F}(T_{t})-\nu[x,\infty)\right|T_{t}\right)+\frac{1}{t}\,\mathbb{E}\left(\,
		1_{\{T_{t}\geq t_{0}\}}\left|\frac{1}{T_{t}}\bar{F}(T_{t})-\nu[x,\infty)\right|T_{t}\right)\\
		&\leq K\frac{1}{t}\mathbb{E}\, T^{2}_{t}+
		B \frac{1}{t}\mathbb{E}(\,
		1_{\{T_{t}\geq t_{0}\}}T_{t})\leq
		K\frac{1}{t}\mathbb{E}\, T^{2}_{t}+
		\frac{B}{t_{0}} \frac{1}{t}\mathbb{E}\,
		(T_{t}^2),
	\end{align*}
	using a Chebyshev upper bound.
	Combining the previous bounds, we have
	\begin{align*}
		\frac{1}{t}\left|\frac{1}{t}\,\mathbb{P}(Z_{t}\geq{}x)-\mathbb{E} Y_{0}\nu[x,\infty)\right|
		&\leq \frac{\nu[x,\infty)}{t}\left|\frac{1}{t}\mathbb{E} T_{t} -\mathbb{E} Y_{0}\right|+
		 \frac{K}{t^{2}}\mathbb{E}\,(T^{2}_{t})+\frac{B}{t_{0}} \frac{1}{t^{2}}\mathbb{E}\,(T_{t}^2)\\
		&\leq  \frac{\nu[x,\infty)}{t^{2}}\int_{0}^{t}
		\left|\mathbb{E} Y_{s} -\mathbb{E} Y_{0}\right| ds+\frac{K}{t^{2}} \mathbb{E} (T_{t}^{2}) +
		\frac{B}{t_{0}} \frac{1}{t^{2}} \mathbb{E} (T_{t}^{2})
	\end{align*}
	Next, (\ref{AN3}) and Jensen's inequality imply that
	\begin{align*}
		 \limsup_{t\to{}0}\frac{1}{t^{2}}\int_{0}^{t}
		\left|\mathbb{E} Y_{s} -\mathbb{E} Y_{0}\right| ds<\infty,
		\quad \limsup_{t\to{}0}\frac{1}{t^{2}}\mathbb{E} (T_{t}^2)\leq \limsup_{t\to{}0}\frac{1}{t}\int_{0}^{t} \mathbb{E} Y^{2}_{s}ds<\infty,
	\end{align*}
	and (\ref{ABTa}) will follow. 	
	In order to show (\ref{ABTb}), consider now
	\[
		G_{x}(t):=\frac{1}{t}\left\{\frac{1}{t} \mathbb{P} (X_{t}\geq{}x) -\nu[x,\infty)\right\}-\frac{d_{2}(x)}{2},
	\]
	and note that, in view of Theorem \ref{ThFLH09}, there exist constants $t_{0}(\varepsilon)>0$ and $K\in(0,\infty)$ such that
	\[
		 \sup_{t>0} |G_{x}(t)|\leq K, \quad \text{and}\quad |G_{x}(t)|<\varepsilon,
	\]
	for any $0<t<t_{0}$.
	As before,
	\begin{align*}
		\frac{1}{t^{2}}\mathbb{P}(Z_{t}\geq{}x)&= \frac{1}{t^{2}} \mathbb{E} \bar{F}(T_{t})=
		\frac{1}{t^{2}}\,\mathbb{E}\left(\left\{\frac{1}{T_{t}}\bar{F}(T_{t})-\nu[x,\infty)\right\}T_{t}\right)+\frac{\nu[x,\infty)}{t^{2}}\,\mathbb{E}\,T_{t}\\
		&=
		 \frac{1}{t^{2}}\,\mathbb{E}\left(G_{x}(T_{t})T_{t}^{2}\right)+\frac{d_{2}(x)}{2t^{2}}\,\mathbb{E}(T_{t}^{2})+\frac{\nu[x,\infty)}{t^{2}}\,\mathbb{E}\,T_{t}\,.
	\end{align*}
	The first term in the last expression can be bounded as follows:
	\begin{align*}
		\left|\frac{1}{t^{2}}\,\mathbb{E}\left(G_{x}(T_{t})T_{t}^{2}\right)\right|&\leq
		\left|\frac{1}{t^{2}}\,\mathbb{E}\left(1_{\{T_{t}<t_{0}\}}G_{x}(T_{t})T_{t}^{2}\right)\right|+
		\left|\frac{1}{t^{2}}\,\mathbb{E}\left(1_{\{T_{t}\geq t_{0}\}}G_{x}(T_{t})T_{t}^{2}\right)\right|\\
		&\leq \frac{\varepsilon}{t^{2}}\mathbb{E}\, (T_{t}^{2})+
		K \frac{1}{t^{2}}\mathbb{E}(\,
		1_{\{T_{t}\geq t_{0}\}}T_{t}^{2}).
	\end{align*}
	Then, it is now clear that we can bound the expression
	\begin{align*}
		D_{t}:=\left|\frac{1}{t^{2}} \mathbb{P} (Z_{t}\geq{}x) -\frac{1}{t}\mathbb{E} Y_{0} \nu[x,\infty)-\frac{\rho d_{2}(x)}{2}
		 - {\frac{\gamma\nu[x,\infty)}{2}}\right|,
	\end{align*}
	as follows
	\begin{align*}
		D_{t}\leq \varepsilon\frac{1}{t^{2}}\mathbb{E}(T_{t}^{2})+K\frac{1}{t^{2}} \mathbb{E}(T^{2}_{t} 1_{\{T_{t}\geq{}t_{0}\}}) +\nu[x,\infty)\left|{\frac{1}{t}\left(\frac{1}{t} \mathbb{E} T_{t} -\mathbb{E} Y_{0}\right) - \frac{\gamma}{2}}\right|+ \frac{|d_{2}(x)|}{2}\left|\frac{1}{t^{2}}\mathbb{E} (T_{t}^{2}) - \rho \right|.
	\end{align*}
	The third term on the right hand side of the above inequality is such that
	\begin{align*}
		 \frac{1}{t}\left(\frac{1}{t} \mathbb{E} T_{t} -\mathbb{E} Y_{0}\right) - \frac{\gamma}{2}=\frac{1}{t^{2}} \int_{0}^{t}s
		 \left\{\frac{1}{s}(\mathbb{E} Y_{s} -\mathbb{E} Y_{0})-\gamma \right\}ds,
	\end{align*}
	which converges to $0$ as $t\to{}0$ due to (iii) in (\ref{AN3}). Hence, using (iv)-(v) in (\ref{AN4}) and
	\[
		 \mathbb{E}(T^{2}_{t} 1_{\{T_{t}\geq{}t_{0}\}})\leq  \mathbb{E}(T^{3}_{t})/t_{0}\leq t^{2}\int_{0}^{t} \mathbb{E} (Y_{s}^{3}) ds/t_{0},
	\]
	we have
	\[
		\limsup_{t\to{}0} D_{t}\leq \varepsilon \rho,
	\]
	which implies (\ref{ABTb}) because $\varepsilon$ is arbitrary.
\end{proof}

\medskip
\nind \begin{proof}[Proof of Lemma \ref{KPN}]
\nind We start by introducing some notation. Suppose that, {under $\mathbb{P}^{*}$}, $X$ has L\'evy-It\^o decomposition
\begin{equation}\label{LevyItoDecmp}
    {X}_{t}= b^{*}t+\int_{0}^{t}\int_{|z|\leq{}1}
    z\,
    \bar\mu^{*}(dz,ds)+
    \int_{0}^{t}\int_{|z|>1} z\,
    \mu^{*}(dz,ds),
\end{equation}
where $\mu^{*}$ is an independent Poisson measure on
$\mathbb{R}\backslash\{0\}\times\mathbb{R}_{+}$ with
mean measure $ \nu^{*}(dz)dt$, and $\bar\mu^{*}(dz,dt):=\mu^{*}(dz,dt)-\nu^{*}(dz)dt$.
Next, for a given fixed  $\varepsilon>0$, we set
\begin{align}\label{TrctedLevy}
    \widetilde{{X}}^{\varepsilon}_{t}:=
    \int_{0}^{t}\int_{\mathbb{R}} z \, {\bf 1}_{\{|z|\geq{}\varepsilon\}} \mu^{*}(dz,ds),\quad\text{ and }\quad
    {X}^{\varepsilon}_{t} := {X}_{t}- \widetilde{{X}}^{\varepsilon}_{t};
\end{align}
hence, $\widetilde{{X}}^{\varepsilon}$ is a  compound Poisson process with intensity
$\lambda_{\varepsilon}:=\nu^{*}(|z|\geq{}\varepsilon)$
and jumps $\{\xi_{i}^{\varepsilon}\}_{i}$ with common distribution
$ {\bf 1}_{|z|\geq{}\varepsilon}\nu^{*}(dz)/\lambda_{\varepsilon}$, while
the remainder process  ${X}^{\varepsilon}$ is a L\'evy process with triplet
$(0,b_{\varepsilon}^{*},
{\bf 1}_{\{|z|\leq{}\varepsilon\}} \nu^{*}(dz))$,
where
\[
    b_{\varepsilon}^{*}:=b^{*}-\int_{|z|\leq{}1}z {\bf 1}_{\{|z|\geq{}\varepsilon\}}\nu^{*}(dz).
\]
Let us fixed $\varepsilon=x/2$.  We first note that
\[
	\mathbb{P}^{*}\left(\widetilde{X}^{\varepsilon}_{t}\geq{}x\right)\leq{} K x^{-Y} t,
\]
for any $t,x>0$ and for some universal constant $K$. Indeed, if we let $N^{\varepsilon}_{t}$ denote the number of jumps before time $t$ of the compound Poisson process $\widetilde{X}^{\varepsilon}$, then we have
\begin{align*}
	\mathbb{P}^{*}\left(\widetilde{X}^{\varepsilon}_{t}\geq{}x\right)
	 &\leq \mathbb{P}^{*}\left(N^{\varepsilon}_{t}\neq{}0\right)=1-e^{-\lambda_{\varepsilon}t}\leq
	 \lambda_{\varepsilon} t=\nu(\{z:|z|\geq x/2\})t\leq C x^{-Y} t.
\end{align*}
We now estimate $\mathbb{P}^{*}\left({X}^{\varepsilon}_{t}\geq{}x\right)$. {First, note} that, due to the symmetry of the L\'evy measure $\nu$,
\begin{align*}
	\mathbb{E}^{*}(X_{t}^{\varepsilon})&=t (b^{*}_{\varepsilon}+\int_{|z|\geq{}1} z {\bf 1}_{\{|z|\leq{}\varepsilon\}} \nu^{*}(dz))=t (b^{*}-\int_{|z|\leq{}1}z {\bf 1}_{\{|z|\geq{}\varepsilon\}}e^{z}\nu(dz)+\int_{|z|\geq{}1} z {\bf 1}_{\{|z|\leq{}\varepsilon\}}e^{z}\nu(dz))\\
	&=t(b+\int_{|z|\leq{}1} z(e^{z}-1)\nu(dz)-\int_{|z|\leq{}1}z {\bf 1}_{\{|z|\geq{}\varepsilon\}}e^{z}\nu(dz)+\int_{|z|\geq{}1} z {\bf 1}_{\{|z|\leq{}\varepsilon\}}e^{z}\nu(dz))\\
	&=t(b+\int_{|z|\leq{}\varepsilon} z(e^{z}-1)\nu(dz))=t(b+\int_{|z|\leq{}x/2} z(e^{z}-1)\nu(dz)).
\end{align*}
Thus,   using concentration inequalities for centered random variable (e.g. \cite{Hou02}, Corollary 1),
for $x>2\mathbb{E} X_{t}^{\varepsilon}$,
\begin{align*}
     \mathbb{P}^{*}({X}_{t}^{\varepsilon}  \geq x) \leq
     \mathbb{P}^{*}({X}_{t}^{\varepsilon} - \mathbb{E}^{*} {X}_{t}^{\varepsilon} \geq x/2)
     \leq {e^{\frac{x}{2\varepsilon}-\left(\frac{x}{2\varepsilon}+\frac{tV_{\varepsilon}^2}{\varepsilon^2}\right)
\log\left(1+\frac{\varepsilon x}{2tV_{\varepsilon}^2}\right)}}
    \leq \left(\frac{2eV_{\varepsilon}^2}{\varepsilon x}\right)^{\frac{x}{2\varepsilon}}t^{\frac{x}{2\varepsilon}}
   \leq \frac{4V_{x/2}^2}{x^{2}} t,
\end{align*}
where $V_{\varepsilon}^2 := {\rm Var}^{*}(X_{1}^{\varepsilon})=\int_{\{|z|\leq{}\varepsilon\}} z^2 \nu^{*}(dz)$. Since $M>1$, there exists a universal constant $K$ such that
\begin{align*}
	\frac{V_{x/2}^2}{x^{2}}&=\frac{C\int_{0}^{x/2} \frac{e^{-(G-1)z}}{z^{1+Y}} z^{2}dz}{x^{2}}+
	\frac{C\int_{0}^{x/2} \frac{e^{-(M-1)z}}{z^{1+Y}} z^{2}dz}{x^{2}}\\
	&\leq \frac{2C\int_{0}^{x/2} z^{1-Y}dz}{x^{2}}=
	\frac{2C (x/2)^{2-Y}}{(2-Y)x^{2}}=K x^{-Y}.
\end{align*}
We conclude that
\(
     \mathbb{P}^{*}({X}_{t}^{\varepsilon}  \geq x) \leq K t x^{-Y}
\)
for $t(b+\int_{|z|\leq{}x/2} z  (e^{z}-1)\nu(dz))< x/2$. This completes the proof, since
\[
	\mathbb{P}^{*}(X_{t}\geq{}x)\leq \mathbb{P}^{*}(X_{t}^{\varepsilon}\geq{}x/2)+\mathbb{P}^{*}(\widetilde{X}^{\varepsilon}_{t}\geq{}x/2)
	\leq K t x^{-Y},
\]
whenever $t(b+\int_{|z|\leq{}x/4} z  (e^{z}-1)\nu(dz))< x/4$.
\end{proof}

\medskip

\noindent\begin{proof}[Proof of Proposition \ref{PATM}]
Without loss of generality, we assume $S_{0}=1$. We break the proof into two parts:

\smallskip
\noindent\textbf{{(1)}}
Let us assume through this part that $(X_{t})_{t}$ is a symmetric CGMY process. Let $b(u):=b+\int_{|z|\leq{}u} z  (e^{z}-1)\nu(dz)$. Obviously,
\begin{align}\label{EIn}
	b(u)\leq |b(u)|\leq|b|+\int_{|z|\leq{}1} |z| |e^{z}-1|\nu(dz)+2\int_{|z|\geq{}1} |z|  e^{z}\nu(dz):=\bar{b}<\infty.
\end{align}
Next, we write
\begin{align}\label{T1AE1}
	\int_{0}^{\infty} e^{-t^{1/Y}u} \mathbb{P}^{*}\left(X_{t}\geq{}t^{1/Y}u\right)du&=\int_{0}^{\infty}{\bf 1}_{\{{u/4}\leq t^{1-1/Y}\bar{b}\}}  e^{-t^{1/Y}u} \mathbb{P}^{*}\left(X_{t}\geq{}t^{1/Y}u\right)du\\
	&+\int_{0}^{\infty}{\bf 1}_{\{{u/4}>t^{1-1/Y}\bar{b}\}} e^{-t^{1/Y}u} \mathbb{P}^{*}\left(X_{t}\geq{}t^{1/Y}u\right)du.\label{T2AE1}
\end{align}
Clearly, $e^{-t^{1/Y}u} \mathbb{P}^{*}\left(X_{t}\geq{}t^{1/Y}u\right)\le 1$, so the first term converges to $0$ as $t\to{}0$ because $Y\in(1,2)$. From the inequality (\ref{EIn}), we have
\[
	{\bf 1}_{\left\{{u/4}>t^{1-1/Y}\bar{b}\right\}}= {\bf 1}_{\left\{t^{1/Y}u>t\bar{b}\right\}}\leq {\bf 1}_{\left\{t^{1/Y}{u/4}>t{b}(t^{1/Y}u/4)\right\}},
\]
and using Lemma \ref{KPN}, we obtain that
\begin{align*}
		{\bf 1}_{\left\{{u/4}>t^{1-1/Y}\bar{b}\right\}} e^{-t^{1/Y}u} \mathbb{P}^{*}\left(X_{t}\geq{}t^{1/Y}u\right)
		&\leq {\bf 1}_{\left\{t^{1/Y}{u/4}>t{b}\left(t^{1/Y}u/4\right)\right\}}e^{-t^{1/Y}u} \mathbb{P}^{*}\left(X_{t}\geq{}t^{1/Y}u\right)\\
		&\leq \min\{K(t^{1/Y}u)^{-Y} t,1\}=\min\{Ku^{-Y} ,1\},
\end{align*}
which is integrable because $Y\in(1,2)$. Hence, we can apply dominated convergence in the second term (\ref{T2AE1}) and,  using Proposition \ref{prop:CGMYConv}, we obtain that
\[
	\lim_{t\to{}0} \int_{0}^{\infty} e^{-t^{1/Y}u} \mathbb{P}^{*}\left(X_{t}\geq{}t^{1/Y}u\right)du =
	\int_{0}^{\infty} \mathbb{P}^{*}(Z\geq{}u)du=\mathbb{E}^{*}(Z_{+}).
\]
This show the result in view of (\ref{CMR})-(\ref{LOP1}).

{
\smallskip
\noindent\textbf{(2)} In this second part, we relax the symmetry restriction.
	The idea is to reduce the problem to the symmetric case by applying a change of probability measure.\footnote{{A similar argument is applied in the proof of Proposition 5-(2) in \cite{Tnkv10} but with a different aim}.}  Concretely, let $\beta:=\frac{M-G}{2}$ and, as in the proof of Proposition \ref{prop:SmallTimeCallsLevy}, define a probability measure $\widehat{\mathbb{P}}$ on $(\Omega, \mathcal{F})$ such that
	\begin{equation}\label{DSM2}
		\widehat{\mathbb{P}}(B)=\mathbb{E}\left(e^{\beta X_{t}} 1_{B} \right)/\mathbb{E}\left(e^{\beta X_{t}}\right),
\end{equation}
for any $B\in\mathcal{F}_{t}$. We can check that, under $\widehat{\mathbb{P}}$,  $(X_{t})_{t}$ is a symmetric CGMY model with $\hat{C}=C$, $\hat{Y}=Y$, and $\hat{G}=\hat{M}=(M+G)/2$. Indeed, it follows that
\begin{align*}
	\hat{\mathbb{E}}\left(e^{iuX_{t}}\right)&=\mathbb{E}\left(e^{(iu+\beta)X_{t}} \right)/\mathbb{E}\left(e^{\beta X_{t}}\right)\\
	&=
	\exp\left[t \,C \Gamma(-Y)\left\{(M-\beta-iu)^Y+(G+\beta+iu)^Y-(M-\beta)^Y-(G+\beta)^Y\right\}+i\hat{b}ut\right].
\end{align*}
Also, assuming $\beta>0$,
\begin{align*}
	\left|{\mathbb{E}}\left(\left(e^{X_{t}}-1\right)_{+}e^{\beta X_{t}}\right)-\mathbb{E}\left(e^{X_{t}}-1\right)_{+}\right|=
	{\mathbb{E}}\left(\left(e^{X_{t}}-1\right)_{+}\left(e^{\beta X_{t}}-1\right)\right)\leq
	{\mathbb{E}}\left(\left(e^{X_{t}}-1\right)\left(e^{\beta X_{t}}-1\right)\right)=O(t),
\end{align*}
since the moment function $\varphi(x):=\frac{1}{\beta}(e^{x}-1)(e^{\beta x}-1)\sim x^{2}$ and Theorem 1.1-(ii) in \cite{FL08} can be applied. If $\beta<0$, then
\begin{align*}
	\left|{\mathbb{E}}\left(\left(e^{X_{t}}-1\right)_{+}e^{\beta X_{t}}\right)-\mathbb{E}\left(e^{X_{t}}-1\right)_{+}\right|=
	{\mathbb{E}}\left(\left(e^{X_{t}}-1\right)_{+}\left(1-e^{\beta X_{t}}\right)\right)\leq
	{\mathbb{E}}\left(\left(e^{X_{t}}-1\right)\left(1-e^{\beta X_{t}}\right)\right)=O(t),
\end{align*}
for the same reason. Then, we only need to consider the asymptotic behavior of ${\mathbb{E}}\left(\left(e^{X_{t}}-1\right)_{+}e^{\beta X_{t}}\right)$ as $t\to{}0$, because $Y \in (1,2)$ so the $O(t)$ terms above are smaller than $O(t^{1/Y})$.  However,
	\[
		{\mathbb{E}}\left(\left(e^{X_{t}}-1\right)_{+}e^{\beta X_{t}}\right)=\mathbb{E}\left(e^{\beta X_{t}}\right)
		\widehat{\mathbb{E}}\left(\left(e^{X_{t}}-1\right)_{+}\right),
	\]
and thus, using the fact that $(X_{t})_{t}$ is symmetric under $\widehat{\mathbb{P}}$ and part (1) in this proof,
\begin{align*}
	\lim_{t\to{}0}t^{-1/Y}{\mathbb{E}}\left(\left(e^{X_{t}}-1\right)_{+}e^{\beta X_{t}}\right)=
	\lim_{t\to{}0}t^{-1/Y}
		\widehat{\mathbb{E}}\left(\left(e^{X_{t}}-1\right)_{+}\right)=\mathbb{E}^{*}(Z_{+}).
\end{align*}
}
\end{proof}

\end{document}